\documentclass[10pt,twocolumn,twoside]{IEEEtran}
\IEEEoverridecommandlockouts
\usepackage{cite}
\usepackage{algorithmic}
\usepackage{textcomp}
\usepackage{xcolor}

\usepackage{graphicx,color}
\usepackage{amsfonts,amsthm,amssymb}
\usepackage{mathrsfs,amsmath,arydshln,latexsym}
\usepackage{cite,url}
\usepackage[bookmarks,colorlinks]{hyperref}
\usepackage[linesnumbered,ruled,vlined]{algorithm2e}
\usepackage{multirow,array,makecell}
\usepackage{stfloats}
\usepackage{balance}
\usepackage{fancyhdr}
\usepackage{bm}
\usepackage{enumitem}
\usepackage{pifont}
\usepackage{subcaption}
\usepackage{nomencl}
\usepackage{colortbl}
\usepackage{diagbox}
\usepackage{balance}

\allowdisplaybreaks[4]

\newtheorem{definition}{Definition}
\newtheorem{theorem}{Theorem}

\newtheorem{corollary}{Corollary}
\newtheorem{lemma}{Lemma}

\begin{document}

	\title{Holographic MIMO Communications\\ with Arbitrary Surface Placements: Near-Field LoS Channel Model and Capacity Limit
	}

	\author{Tierui Gong,~\IEEEmembership{Member,~IEEE,} 
		Li Wei, 
		Chongwen Huang,~\IEEEmembership{Member,~IEEE,} 
		Zhijia Yang, \\
		Jiguang He,~\IEEEmembership{Senior Member,~IEEE,}
		Mérouane Debbah,~\IEEEmembership{Fellow,~IEEE,} 
		and Chau Yuen,~\IEEEmembership{Fellow,~IEEE} 
		\vspace{-0.5cm}
		\thanks{This work of C. Yuen was supported in part by the Ministry of Education, Singapore, under its MOE Tier 2 (Award number MOE-T2EP50220-0019), and also by the Science and Engineering Research Council of A*STAR (Agency for Science, Technology and Research) Singapore, under Grant No. M22L1b0110.
		The work was supported by the China National Key R\&D Program under Grant 2021YFA1000500 and 2023YFB2904800, National Natural Science Foundation of China under Grant 62331023 and 62101492, Zhejiang Provincial Natural Science Foundation of China under Grant LR22F010002, Zhejiang University Global Partnership Fund, and Zhejiang University Education Foundation Qizhen Scholar Foundation.
		Part of this article was accepted by ICC workshop 2023 \cite{Gong2023Generalized}
		}
		\thanks{T. Gong, L. Wei and C. Yuen are with School of Electrical and Electronics Engineering, Nanyang Technological University, Singapore 639798 (e-mail: trgTerry1113@gmail.com, l\_wei@ntu.edu.sg, chau.yuen@ntu.edu.sg).}
		\thanks{C. Huang is with College of Information Science and Electronic Engineering, Zhejiang University, Hangzhou 310027, China, and Zhejiang Provincial Key Laboratory of Info. Proc., Commun. \& Netw. (IPCAN), Hangzhou 310027, China. (e-mail: chongwenhuang@zju.edu.cn).
		}
		\thanks{Z. Yang is with State Key Laboratory of Robotics, Shenyang Institute of Automation, Chinese Academy of Sciences, Shenyang 110016, China, with Key Laboratory of Networked Control Systems, Chinese Academy of Sciences, Shenyang 110016, China, and also with Institutes for Robotics and Intelligent Manufacturing, Chinese Academy of Sciences, Shenyang 110169, China. (e-mail: yang@sia.ac.cn).}
		\thanks{J. He is with Technology Innovation Institute, Masdar City 9639, Abu Dhabi, United Arab Emirates (e-mail: jiguang.he@tii.ae).}
		\thanks{M. Debbah is with KU 6G Research Center, Khalifa University of Science and Technology, P O Box 127788, Abu Dhabi, UAE (email: merouane.debbah@ku.ac.ae) and also with CentraleSupelec, University Paris-Saclay, 91192 Gif-sur-Yvette, France.}
	\vspace{-0.1cm}
}

\maketitle

\begin{abstract}
	Envisioned as one of the most promising technologies, holographic multiple-input multiple-output (H-MIMO) recently attracts notable research interests for its great potential in expanding wireless possibilities and achieving fundamental wireless limits. Empowered by the nearly continuous, large and energy-efficient surfaces with powerful electromagnetic (EM) wave control capabilities, H-MIMO opens up the opportunity for signal processing in a more fundamental EM-domain, paving the way for realizing holographic imaging level communications in supporting the extremely high spectral efficiency and energy efficiency in future networks. 
	In this article, we propose a generalized EM-domain near-field channel modeling and study its capacity limit of point-to-point H-MIMO systems that equips arbitrarily placed surfaces in a line-of-sight (LoS) environment. Two effective and computational-efficient channel models are established from their integral counterpart, where one is with a sophisticated formula but showcases more accurate, and another is concise with a slight precision sacrifice. Furthermore, we unveil the capacity limit using our channel model, and derive a tight upper bound based upon an elaborately built analytical framework. 
	Our result reveals that the capacity limit grows logarithmically with the product of transmit element area, receive element area, and the combined effects of $1/{{d}_{mn}^2}$, $1/{{d}_{mn}^4}$, and $1/{{d}_{mn}^6}$ over all transmit and receive antenna elements, where $d_{mn}$ indicates the distance between each transmit element $n$ and receive element $m$. Particularly, $1/{{d}_{mn}^6}$ dominates in the near-field region whereas $1/{{d}_{mn}^2}$ dominates in the far-field region.
	Numerical evaluations validate the effectiveness of our channel models, and showcase the slight disparity between the upper bound and the exact capacity, which is beneficial for predicting practical system performance. 
\end{abstract}

\begin{IEEEkeywords}
	Holographic MIMO, channel modeling, capacity limit, near-field LoS, arbitrary surface placement.
\end{IEEEkeywords}

\section{Introduction}

The sixth-generation (6G) wireless networks are envisioned to provide immersive holographic-type communications (communications supporting multisensory XR, holographic videos, etc.) in supporting the extremely large amount of data traffic and a multitude of various upper-layer applications with a special requirement on high energy efficiency \cite{Saad2020Vision}.
Massive multiple-input multiple-output (M-MIMO) technology with tens to hundreds of antenna elements was proposed as a critical enabler for the fifth-generation (5G) wireless networks \cite{Rusek2013Scaling}, and its enhanced version, such as ultra-massive MIMO (UM-MIMO) with thousands of antenna elements, was envisioned as a potential technology for 6G to meet the extreme data rate requirement \cite{Akyildiz2018Combating}. This is not sustainable in terms of energy efficiency when the number of antennas scales up to be extremely large. The infeasibility is mainly caused by using the conventional power-hungry and cost-inefficient antenna technology, where the fully-digital hardware structure, requiring a dedicated radio-frequency chain for each antenna element, is applied \cite{Xiao2017Millimeter}. Even though this problem can be relieved to some extent by using hybrid analog-digital structures \cite{Molisch2017Hybrid, Gong2020RF}, it is still not extendable for supporting future 6G due to the ever increasing cost and power consumption for systems with extremely large number of conventional antenna elements. 

Holographic MIMO (H-MIMO) has been envisioned as a promising technology capable of revolutionizing the traditional M-MIMO and UM-MIMO, as well as facilitating the developments of various aspects of 6G \cite{Gong2023Holographic}. H-MIMO exploits new antenna technologies, such as metasurface antennas \cite{Shlezinger2021Dynamic}, offering at least retained performance while reducing power consumption and cost. In particular, the antenna array used in H-MIMO is generally an antenna surface consisting of a certain number of feed ports, and densely packed elements (appearing infinitesimally small) over a substrate, as shown in Fig. \ref{fig:HMIMOPrinciple}. 
H-MIMO generally follows the holographic imaging principle, where one or more reference waves, loaded by the feed ports, propagate along the substrate and successively reach the antenna elements, thereby, exciting all tunable elements to form a specific surface current distribution in order to generate the intended object waves \cite{Gong2023Holographic, Huang2020Holographic}. One of the typical realizations is using artificial metamaterials/metasurfaces, allowing electromagnetic (EM) wave manipulations with low-cost, low power consumption, and the generation of various expected EM responses. In general, one of the prominent features of H-MIMO is that the antenna surface appears to be nearly continuous with a massive (possibly infinite) number of antenna elements integrated into a compact space. 
In addition, another conspicuous characteristic of H-MIMO is its electronically large surface area, which is used for combating large path-losses at high frequencies, and is made feasible owing to the low-cost and low power consumption nature. These new features, not found in conventional M-MIMO systems, bring new challenges and opportunities for H-MIMO communications. 

\begin{figure}[t!]
	\centering
	\includegraphics[height=5.6cm, width=6.8cm]{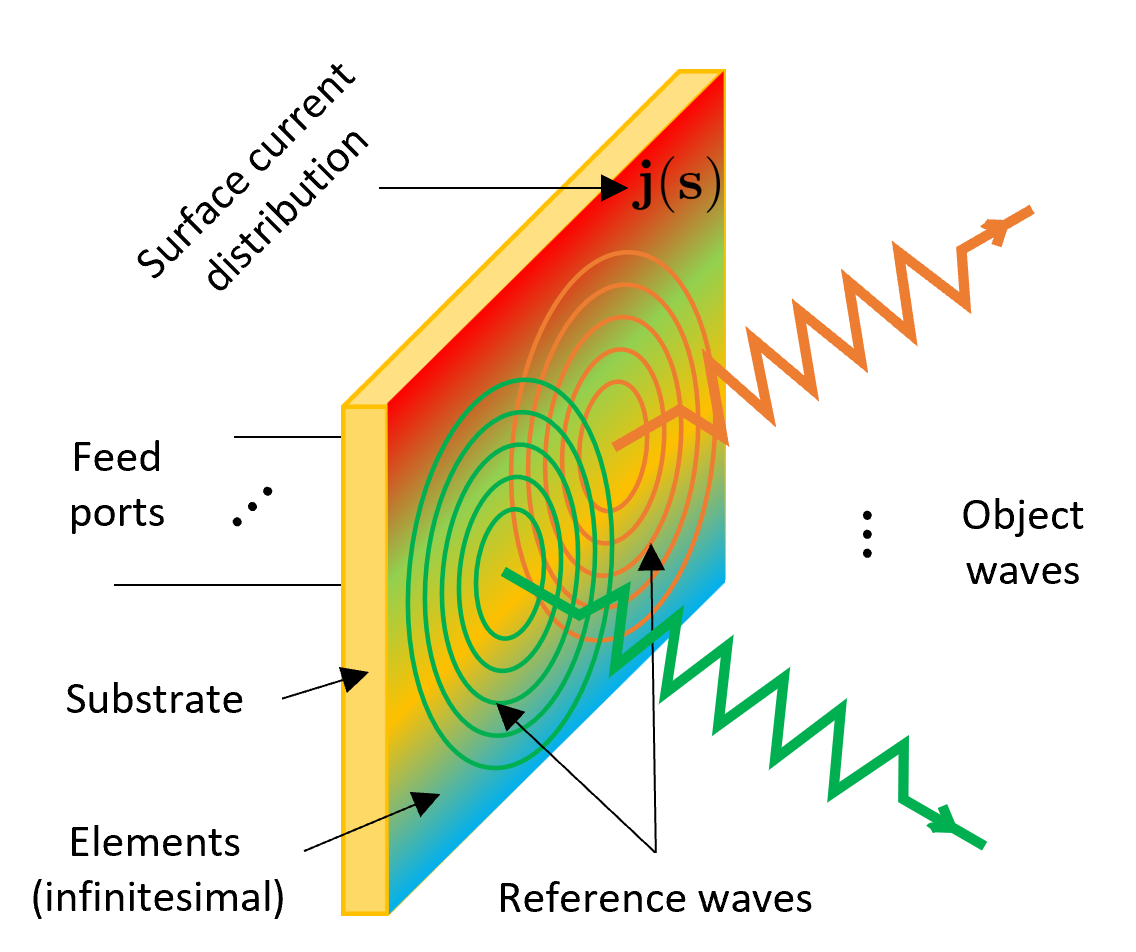}
	\caption{The main components and working principle of H-MIMO systems.}
	\label{fig:HMIMOPrinciple}
\end{figure}

Because of the nearly continuous surface, and the inherently powerful capability in EM wave control, H-MIMO can record the spatial information of EM waves to the utmost degree, and manipulate EM waves to obtain almost arbitrary degrees of freedom (DoF), which can promisingly approach the fundamental limit of the wireless environment \cite{Dardari2021Holographic}. Furthermore, the significantly less than half-a-wavelength spacing between elements inevitably causes strong mutual couplings among antenna elements, which is envisioned to be exploited for achieving super-directivity \cite{Marzetta2019Super}. To unveil the fundamental limit and study the mutual coupling effect, one of the critical challenges is to properly model the H-MIMO system from a more fundamental EM perspective. 
In addition, the electrically large surface feature of H-MIMO will facilitate near-field communications, where the communication distance falls within the Rayleigh distance \cite{Cui2022Near}. As opposed to far-field communications prevalent in conventional M-MIMO systems, which merely exploits the angle information, the near-field communications induced by H-MIMO can leverage both angle and distance information in assisting transmissions. This is because the signals experienced over all element pairs between the transceiver have different magnitudes and phases that are distance and angle dependent due to the spherical wavefront in the near-field region. Consequently, H-MIMO systems will significantly enlarge the DoF in near-field communications. In this regard, H-MIMO paves the way for holographic imaging level near-field communications \cite{Gong2023Holographic}. Generally speaking, both EM-domain modeling and analysis, especially for near-field cases, are critical to H-MIMO, where the EM-domain channel modeling and capacity limit analysis are of great significance.



\subsection{Prior Studies}

In channel modeling, the new features of H-MIMO inevitably introduce significant changes that should be addressed from a fundamental EM perspective. Specifically, the EM wave field, the actual transmission carrier in H-MIMO communications, is a spatial vector, and the communication distance tends to be in the near-field region, which allows the realization of communications with more polarization.  
Conventional channel models, such as the classic Rayleigh fading channel model \cite{Tse2005Fundamentals} and its correlated version \cite{Bjornson2018Massive}, as well as the cluster-based channel model \cite{Heath2016Overview,Cheng2022RIS}, are generally built for far-field scenarios and are all based on mathematical abstractions, depicting the wireless channel via mathematical representations, while ignoring the physical phenomena of EM wave propagation. This is, however, insufficient to describe the wireless channel for H-MIMO communications. As shown recently in \cite{Pizzo2020Spatially, Pizzo2022Fourier,WeiLi2022Multi-user}, the authors describe a wireless channel following EM principles, where they studied the small-scale fading for scalar wave field in far-field scenarios. 
As antenna surfaces tend to be large, the near-field line-of-sight (LoS) channel should be considered. In recent studies \cite{Zhang2022Beam, Cui2022Channel}, the near-field LoS channel is described using a spherical wavefront propagation model, which is more of a mathematical abstraction without emphasis on EM propagation phenomena, e.g. vector wave field and multi-polarization. 
The above channel models undoubtedly fail to support the realistic vector wave field scenario, especially operating in the near-field region where the interactions among EM wave fields are abundant and complicated.
Going one step further, recent works \cite{Gong2023Transmit, Wei2022Tri} proposed the EM-compliant near-field LoS channel models for H-MIMO, respectively, which are capable of depicting the vector wave field case. However, the former work focuses on deriving a measurement-efficient model with high flexibility and mathematical tractability, whereas sacrificing the depiction accuracy to some extent. The latter work only considers the parallel placement of surfaces, leading to the failure in capturing channel responses for arbitrary surface placements, which is the general case in practical deployments.

Existing capacity limit analyses are based on mathematically abstracted channel models and/or assume a parallel placement of antenna surfaces.
The works \cite{Hu2018LIS, Jesus2020Near, Lu2022Communicating} considered large intelligent surfaces or extremely large MIMO aided multiple single-antenna user systems, and studied the system performance under near-field LoS channel using a spherical wavefront propagation channel model. These works perform analyses based on mathematical abstractions not only in channel modeling but also in system modeling \cite{Jesus2020Near, Lu2022Communicating}. 
In addition, a few studies perform analyses with the EM-domain modeling incorporated. Specifically, 
\cite{WeiLi2022Multi-user} considered the far-field small-scale fading in a multi-user H-MIMO system, and analyzed the capacity of different transmission schemes by using an established EM-domain channel model. The applied channel model is a multi-user extension of \cite{Pizzo2020Spatially,Pizzo2022Fourier}, and a comprehensive study based upon a rigorous EM derivation can be found in \cite{Poon2005Degrees}. Another recent work \cite{Wang2022HIHO}, taking advantage of the EM-domain system modeling for continuous surfaces while using a simple spherical wavefront propagation channel model, analyzed the performance of a point-to-point H-MIMO system in the near-field LoS environment. Even though existing studies unveil the capacity limit to some extent, it is worth noting that the study on capacity limit of H-MIMO from the EM perspective, especially for the near-field case, is still on the way.
In addition to the capacity limit analysis, several efforts have been made to introduce the EM-domain knowledge into the communication analysis, such as DoF and power coupling between transceivers. Particularly, \cite{Dardari2020Communicating2} established an EM-domain system model for a point-to-point H-MIMO system in the near-field LoS channel, where the power gain and DoF were analyzed and several useful results were provided for certain special cases. Some of the theoretical foundations of EM-domain system modeling and analysis can be dated back to works \cite{Piestun2000Electromagnetic, Miller2000Communicating} for optical systems. 
These EM-domain driven studies mainly focus on H-MIMO systems with the parallel surface placement, which has a limited generalization to practical scenarios with arbitrary surface placements. Moreover, we emphasize that these EM-driven studies, e.g., \cite{Poon2005Degrees, Dardari2020Communicating2}, mainly focus on the extreme scenario (i.e., totally continuous surface) to obtain theoretical insights, which is unlike the realistic, discrete surface with densely packed elements (nearly continuous surface) in system design and signal processing. For instance, multiple integral operations possibly emerge in modeling the totally continuous surfaces, which incurs high computational complexity and lacks mathematical insights.

In summary, most of the prior studies focus on channel modeling via the mathematical abstraction \cite{Tse2005Fundamentals, Bjornson2018Massive, Heath2016Overview, Cheng2022RIS, Zhang2022Beam, Cui2022Channel}, where only limited channel features can be captured whereas others are ignored, such as the vector wave field feature and the multi-polarization feature. These issues are also suffered by several EM-domain channel models \cite{Pizzo2020Spatially, Pizzo2022Fourier,WeiLi2022Multi-user} due to their consideration of a relatively simple scalar wave field case. 
The defects	become more predominant in the near-field LoS channel modeling. 
In addition to the channel modeling, the existing works, focusing on capacity limit analysis, rely on either mathematically abstracted channel models \cite{Hu2018LIS, Jesus2020Near, Lu2022Communicating} or some simplified scenario (e.g., scalar wave field) \cite{WeiLi2022Multi-user, Wang2022HIHO}, thereby failing to attain the true ultimate limit of the physical channels.
Even though a few recent works investigate EM-domain channel modeling \cite{Wei2022Tri} and fundamental limits (DoF and power coupling) \cite{Dardari2020Communicating2} for vector wave fields, they are actually restricted to the parallel surface placement condition.


\subsection{Contributions}

To mitigate the gaps, we focus on a point-to-point H-MIMO system with arbitrary surface placements, and propose to depict the whole system from an EM perspective. More importantly, we establish the generalized EM-domain near-field LoS channel models and unveil the capacity limit on this basis. The details of our contributions are listed as follows. 
\begin{itemize}
	\item 
	We present an effective approach for characterizing the arbitrariness of surface placements, capable of controlling the location and orientation of any surface (and its elements) as well as two different surface (element) shapes. This approach also allows us to represent arbitrary point on the surface by two defined unit vectors, and to denote the $z$ orientation of a surface point as a function of its $x$ orientation and $y$ orientation, facilitating the near-field LoS channel modeling.
	
	\item 
	We study a generalized EM-domain channel modeling for near-field LoS scenarios, by employing our characterization approach for arbitrary surface placements, and utilizing the Taylor series expansion and reasonable approximations to tackle the coupling and computationally-infeasible problems caused by multiple integrals. Accordingly, a sophisticated but more accurate coordinate-dependent channel model (CD-CM), depending on the absolute location of antenna elements, is first established. With moderate approximations, the CD-CM is further simplified to a concise coordinate-independent channel model (CI-CM) with a slight precision sacrifice, which is merely determined by the relative distances and their directions. The proposed channel models are capable of capturing the inherent physical channel features, such as vector wave field and multi-polarization that are ignored in mathematically abstracted channel models.
	
	\item 
	We further examine the capacity limit of H-MIMO systems using the established CI-CM in a rigorous derivation. The analysis is performed based upon an analytical framework, incorporating formally defined transmit and receive patterns and a resulting convenient channel decomposition. These patterns indicate orthogonal bases capable of depicting any current distribution, electric field and the tensor Green's function based wireless channel. 
	This framework facilitates our analysis in deriving a tight closed-form capacity upper bound. It reveals that the capacity grows logarithmically with the product of transmit element area, receive element area, and combined effects of $1/{{d}_{mn}^2}$, $1/{{d}_{mn}^4}$, and $1/{{d}_{mn}^6}$ over all transmit and receive antenna elements, where $d_{mn}$ is the distance between each transmit element $n$ and receive element $m$. Particularly, $1/{{d}_{mn}^6}$ dominates in the near-field region whereas $1/{{d}_{mn}^2}$ dominates in the far-field region.
	
	\item 
	We finally evaluate the established channel models and capacity limit through extensive numerical simulations. The results validate the feasibility of our channel models and demonstrate the H-MIMO capacity limit, offering various insights for system designs. 
\end{itemize}

\subsection{Organizations}

The rest of this article is organized as follows: In Section 
\ref{SectionSM}, we establish an EM-domain system model. In Section \ref{SectionCM}, near-field LoS channel models are established based on an effective representation approach for surface placements. We further present an analytical framework and perform analysis on the capacity limit in Section \ref{SectionCL}. Extensive numerical results are provided in Section \ref{SectionNR}, and conclusions are made in Section \ref{SectionCON}.

\section{H-MIMO System Model}
\label{SectionSM}

In the following of this section, we focus on providing the EM-domain system model and characterizing the arbitrariness of surface placements. Particularly, we consider an H-MIMO communication system consisting of one transmitter (TX) communicating with one receiver (RX). Both TX and RX are equipped with a massive (possibly infinite) number of antenna elements integrated into a compact space.
For convenience, we denote by $N = N_{h} \times N_{v}$ the overall number of antenna elements of TX, consisting of $N_{h}$ and $N_{v}$ antenna elements in the horizontal direction and the vertical direction, respectively. Likewise, the RX has an overall $M = M_{h} \times M_{v}$ antenna elements.
Besides, we assume that the antenna elements are densely packed over the surface with a unified element spacing, and each TX element has an area of $s_{T}$ with a horizontal length $l_{T}^{h}$ and a vertical length $l_{T}^{v}$. 
Accordingly, the overall surface area of TX can thus be derived as $A_{T} = N s_{T}$. 
This is extended in a similar way to the RX with element area $s_{R}$, horizontal length $l_{R}^{h}$, vertical length $l_{R}^{v}$, and an overall surface area $A_{R} = M s_{R}$. Moreover, we depict the whole system in Cartesian coordinates, and consider an arbitrary placement of antenna surfaces, as shown in Fig. \ref{fig:TX-RX-System}, which is more general than the commonly considered parallel placement \cite{Hu2018LIS,Dardari2020Communicating2,WeiLi2022Multi-user,Wang2022HIHO,Wei2022Tri}.

\begin{figure}[t!]
	\centering
	\includegraphics[height=4.6cm, width=8.0cm]{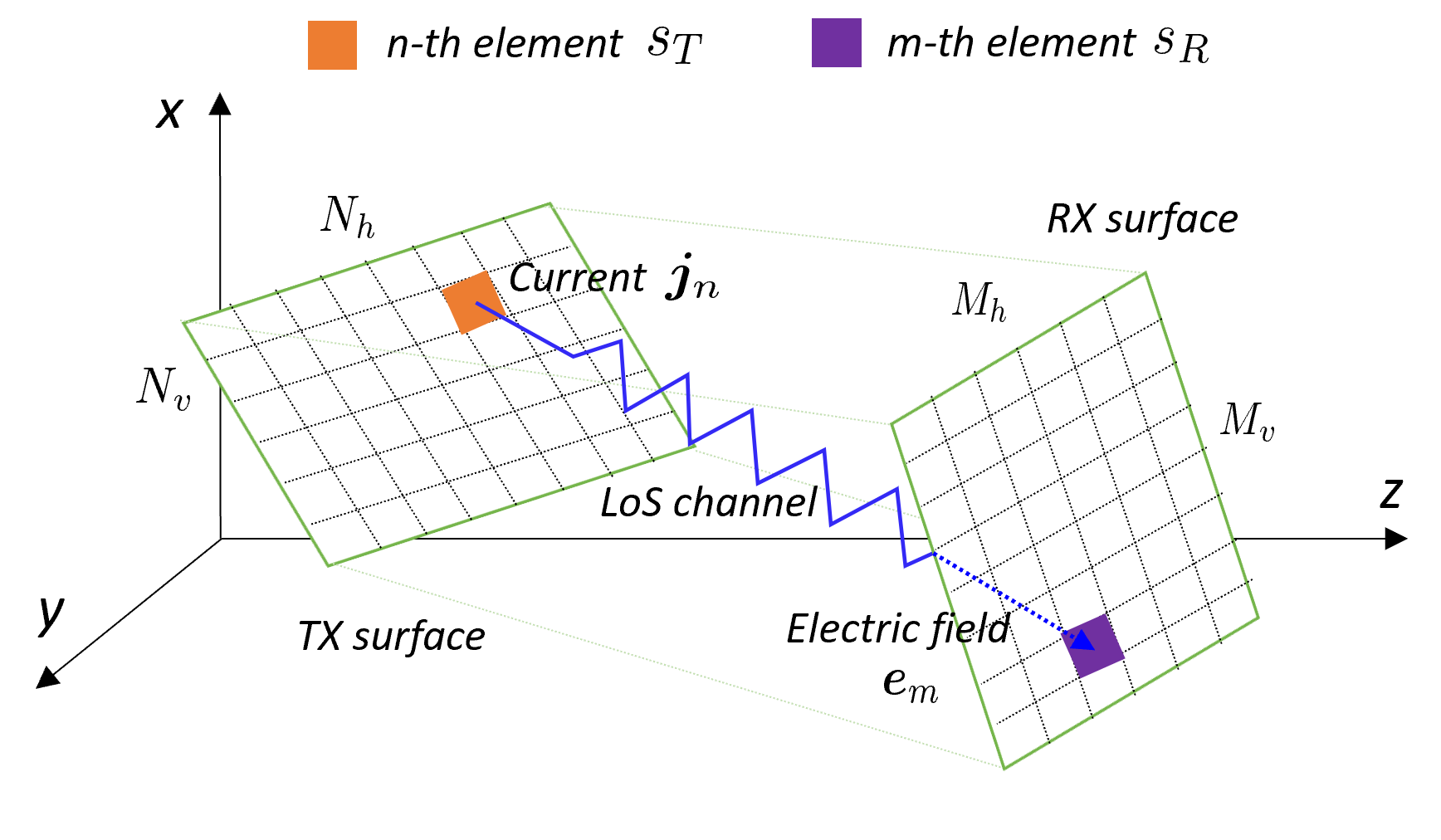}
	\caption{Cartesian coordinates of the TX/RX surfaces with arbitrary surface placements.}
	\label{fig:TX-RX-System}
	\vspace{-0.5em}
\end{figure}

\subsection{EM-Domain System Model}

As previously described, H-MIMO is capable of controlling the EM waves with an unprecedented flexibility, whose communication model needs to be depicted from the EM perspective. In general, one can employ the current distribution and the electric field for demonstrating the transmit and receive relation. Specifically, the electric field $\bm{e}(\bm{r}) \in \mathbb{C}^{3 \times 1}$ measured at a certain point location $\bm{r}$ is an instance excited by the current distribution $\bm{j}(\bm{t}) \in \mathbb{C}^{3 \times 1}$ imposed at a different point location $\bm{t}$ after passing through the wireless channel $\bm{H}_{\bm{rt}} \in \mathbb{C}^{3 \times 3}$. Since the whole system is located in the three-dimensional space, the current distribution and the electric field, oriented in arbitrary directions, can be represented by three sub-components as $\bm{j}(\bm{t}) = [j^x, j^y, j^z]^{T}$ and $\bm{e}(\bm{r}) = [e^x, e^y, e^z]^{T}$, respectively, corresponding to $xyz$-axis. Their connecting wireless channels $[\bm{H}_{\bm{rt}}]_{pq} = H_{pq}$, $p,q \in \{x, y, z\}$ map current distribution sub-components to electric field sub-components from the $q$-orientation to the $p$-orientation. 

To demonstrate this, we focus on a current distribution $\bm{j}(\bm{t}_n)$, imposed at location $\bm{t}_{n} = [x_n, y_n, z_n]^{T}$ of the TX surface. Due to the imposed current distribution, the incurred radiated electric field, measured at an arbitrary location $\bm{r}_{m} = [x_m, y_m, z_m]^{T}$ within the region of the $m$-th antenna element of RX surface, is given by \cite{Chew1999Waves}
\begin{align}
	\label{eq:P2P-ElectricField-Current}
	{\bm{e}}\left( {{{\bm{r}}_m}} \right) = \frac{\eta}{2 \lambda} \int_{s_T} {{\bm{G}}\left( {{{\bm{r}}_m},{{\bm{t}}_n}} \right)} {\bm{j}}\left( {{{\bm{t}}_n}} \right) {\rm{d}}{{\bm{t}}_n} + \bm{w} \left( {{{\bm{r}}_m}} \right),
\end{align}
where $\eta$ is the intrinsic impedance of free-space (i.e., $376.73$ $\Omega$), and $\lambda$ is the free-space wavelength; $\bm{w} \left( {{{\bm{r}}_m}} \right)$ indicates the noise measured at ${{{\bm{r}}_m}}$, following $\mathcal{CN}(0, \sigma_w^2)$; ${{\bm{G}}\left( {{{\bm{r}}_m},{{\bm{t}}_n}} \right)} \in \mathbb{C}^{3 \times 3}$ is the tensor Green's function, given by
\begin{align}
	\label{eq:GreenFunc}
	{\bm{G}}\left( {{{\bm{r}}_m},{{\bm{t}}_n}} \right)
	= \frac{-i}{{4\pi d_{mn}}} \left( {{{\bm{I}}_3} + \frac{1}{{k_0^2}} \bm{\nabla} {\bm{\nabla} ^T}} \right) {{e^{i{k_0} d_{mn}}}}, 
\end{align}
where $i^2 = -1$, $\bm{I}_3$ is the $3 \times 3$ identity matrix; $k_0$ denotes the free-space wavenumber, $k_0 = \frac{2 \pi}{\lambda}$; $\bm{\nabla} = [\frac{\partial}{\partial x}, \frac{\partial}{\partial y}, \frac{\partial}{\partial z}]^{T}$ represents the gradient operator; 
$d_{mn} \triangleq \left \| \bm{d}_{mn} \right \|_{2} \triangleq \left \| {{{\bm{r}}_m} - {{\bm{t}}_n}} \right \|_{2}$ with $\| \cdot \|_{2}$ denoting the $l_2$ norm of a vector.
On this basis, the overall electric field received by the $m$-th receive antenna element can be obtained as a sum of ${\bm{e}} \left( {{{\bm{r}}_m}} \right)$ over the element area $s_{R}$, i.e.,  
\begin{align}
	\nonumber
	{\bm{e}}_m 
	= \int\limits_{s_R} {\bm{e}} \left( {{{\bm{r}}_m}} \right) {\rm{d}}{{\bm{r}}_m} 
	&= \frac{\eta}{2 \lambda}  \int_{s_R} \int_{s_T} {{\bm{G}}\left( {{{\bm{r}}_m},{{\bm{t}}_n}} \right)} {\bm{j}}\left( {{{\bm{t}}_n}} \right)  {\rm{d}}{{\bm{t}}_n} {\rm{d}}{{\bm{r}}_m} \\
	\label{eq:MP2MP-ElectricField-Current}
	&+ s_{R} \bm{w}_m,
\end{align}
where we assume that the noise is uniformly distributed over the whole $s_R$ area, namely, $\bm{w} \left( {{{\bm{r}}_m}} \right) \triangleq \bm{w}_m$. 
We emphasize that \eqref{eq:MP2MP-ElectricField-Current} builds the input-output model of each $mn$-pair, which can depict any non-uniform current distributions. However, it is too complicated to integrate this model for depicting the whole H-MIMO system. 
To further simplify this model, it is reasonable to assume that the current distribution imposed to the antenna element is uniformly distributed over each area $s_{T}$. 
namely, $\bm{j}(\bm{r}_n) \triangleq \bm{j}_n$. Therefore, \eqref{eq:MP2MP-ElectricField-Current} can be simplified to
\begin{align}
	\label{eq:MP2MP-ElectricField-Current-1}
	{\bm{e}}_m = \bm{H}_{mn} \bm{j}_n + s_{R} \bm{w}_m,
\end{align}
where the near-field LoS channel $\bm{H}_{mn}$ is defined as
\begin{align}
	\label{eq:MP2MP-Channel}
	\bm{H}_{mn} = \frac{\eta}{2 \lambda} \int_{s_R} \int_{s_T} {{\bm{G}}\left( {{{\bm{r}}_m},{{\bm{t}}_n}} \right)} {\rm{d}}{{\bm{t}}_n} {\rm{d}}{{\bm{r}}_m},
\end{align}
which is an integral form channel model (INT-CM).
It is observed that the simplified input-output model \eqref{eq:MP2MP-ElectricField-Current-1} provides an explicit representation of the communication link between the $mn$-pair.

We extend the basic model \eqref{eq:MP2MP-ElectricField-Current-1} to the whole H-MIMO system by embedding $\bm{H}_{mn}$ as the $(m, n)$-th element of the overall H-MIMO channel matrix. As such, the relation between the imposed current distributions of all TX antenna elements to the measured electric fields of all RX antenna elements can be obtained as
\begin{align}
	\label{eq:P2P-HMIMO-1}
	\bm{e} = \bm{H} \bm{j} + s_{R} \bm{w},
\end{align}
where $\bm{e} = [{\bm{e}}_1, {\bm{e}}_2, \cdots, {\bm{e}}_M]^{T}$ represents the vector of all measured electric fields with each element given by $\bm{e}_{m} = [e_{m}^{x}, e_{m}^{y}, e_{m}^{z}]^{T}$; $\bm{w} = [{\bm{w}}_1, {\bm{w}}_2, \cdots, {\bm{w}}_N]^{T}$ denotes the additive white Gaussian noise vector, where $\bm{w}_{n} = [w_{n}^{x}, w_{n}^{y}, w_{n}^{z}]^{T}$; $\bm{j} = [{\bm{j}}_1, {\bm{j}}_2, \cdots, {\bm{j}}_N]^{T}$ is the vector of imposed current distributions with elements $\bm{j}_{n} = [j_{n}^{x}, j_{n}^{y}, j_{n}^{z}]^{T}$; and
\begin{align}
	\label{eq:P2P-HMIMO-ChannelModel-1}
	\bm{H} = \left[ {\begin{array}{*{20}{c}}
			{{{\bm{H}}_{11}}}&{{{\bm{H}}_{12}}}& \cdots &{{{\bm{H}}_{1N}}}\\
			{{{\bm{H}}_{21}}}&{{{\bm{H}}_{22}}}& \cdots &{{{\bm{H}}_{2N}}}\\
			\vdots & \vdots & \ddots & \vdots \\
			{{{\bm{H}}_{M1}}}&{{{\bm{H}}_{M2}}}& \cdots &{{{\bm{H}}_{MN}}}
	\end{array}} \right] 
\end{align}
with each element matrix given by \eqref{eq:MP2MP-Channel}.

\section{EM-Domain Near-Field LoS Channel Modeling}
\label{SectionCM}

As seen from the channel expression derived in previous section that the near-field LoS channel is depicted by integrals of the tensor Green's function, which is implicit and computationally infeasible, especially for obtaining the channel matrix \eqref{eq:P2P-HMIMO-ChannelModel-1} of an H-MIMO system with nearly infinite antenna elements. 
To mitigate the gap, in the following of this section, we propose to derive an explicit expression of the channel by eliminating the integral operators. 
Particularly, we exploit the Taylor series expansion and propose an effective approach for surface placement representation to facilitate the integral solution. 
As a consequence, a CD-CM, relying on the absolute coordinates of antenna element locations, is first derived, which shows a good approximation to the INT-CM (demonstrated in numerical evaluations below). Using a reasonable approximation, this model is further simplified to a concise CI-CM with an acceptable precision sacrifice. This concise model only relies on the relative locations of antenna elements, paving the way for its wide applications.

The double integrals of \eqref{eq:MP2MP-Channel} over TX and RX element areas can be expressed as multiple integrals over three spatial orientations ($x$, $y$ and $z$ orientations). It is quite challenging to eliminate those multiple integrals due to the complicated coupling among them. In order to decompose the strong coupling, we present the following representation of an arbitrary point located on the TX (RX) surface via using the center coordinates. 
Since the TX (RX) surface is divided into discrete antenna elements, an arbitrary point located on the TX (RX) surface belongs to a certain antenna element.
Without loss of generality, we select the $n$-th antenna element of TX and the $m$-th antenna element of RX, and denote their center coordinates by $\bar{\bm{t}}_{n} = [\bar{x}_n, \bar{y}_n, \bar{z}_n]^{T}$ and $\bar{\bm{r}}_{m} = [\bar{x}_m, \bar{y}_m, \bar{z}_m]^{T}$, respectively.
Let $\bm{t}_{n}$ and $\bm{r}_{m}$ be the arbitrary points belonging to the $n$-th TX antenna element and the $m$-th RX antenna element, respectively. 
It is obvious that $\bm{t}_{n}$ and $\bm{r}_{m}$ lie in the nearby regions of $\bar{\bm{t}}_{n}$ and $\bar{\bm{r}}_{m}$, and are within $s_T$ and $s_R$, respectively, as demonstrated in Fig. \ref{fig:ElementRepresentation}. Therefore, we have 
\begin{align}
	\label{eq:AnyPointOnTE}
	\bm{t}_{n} &= \bar{\bm{t}}_{n} + \Delta \bm{t}_{n}, \\
	\label{eq:AnyPointOnRE}
	\bm{r}_{m} &= \bar{\bm{r}}_{m} + \Delta \bm{r}_{m}, 
\end{align}
where ${\Delta {\bm{t}}_n} \triangleq [{\Delta {x}_n} \; {\Delta {y}_n} \; {\Delta {z}_n}]^{T}$ and ${\Delta {\bm{r}}_m} \triangleq [{\Delta {x}_m} \; {\Delta {y}_m} \; {\Delta {z}_m}]^{T}$ belong to $s_T$ and $s_R$, respectively. 
Accordingly, any arbitrary point within the TX/RX surface can be demonstrated by the center coordinates of a certain antenna element. 

\subsection{From Double Integral to Separated Single Integrals}

Going one step further, an explicit expression of the tensor Green's function is first employed, which is given by \cite{Arnoldus2001Representation}
\begin{align}
	\nonumber
	{\bm{G}}\left( {{{\bm{r}}_m},{{\bm{t}}_n}} \right)
	&= \frac{-i e^{i{k_0} d_{mn}}}{{4\pi d_{mn}}} \left[ \left( {1 + \frac{i}{{{k_0}d_{mn}}} - \frac{1}{{k_0^2{d_{mn}^2}}}} \right) {{\bm{I}}_3} \right. \\
	\label{eq:GreenFunc-1}
	&+ \left. \left( {\frac{3}{{k_0^2{d_{mn}^2}}} - \frac{{3i}}{{{k_0}d_{mn}}} - 1} \right) \frac{ \bm{d}_{mn} \bm{d}_{mn}^{T} } {{{d_{mn}^2}}} \right],
\end{align}
where $d_{mn} = \left \| \bm{d}_{mn} \right \| = \left \| {{\bm{r}}_{m} - {\bm{t}}_{n}} \right \|_{2}$. 
Applying \eqref{eq:GreenFunc-1} to \eqref{eq:MP2MP-Channel}, we see that it is still difficult to obtain an explicit expression due to the complicated coupling among multiple integrals.
As previously shown in \eqref{eq:AnyPointOnTE} and \eqref{eq:AnyPointOnRE}, $\bm{t}_{n}$ and $\bm{r}_{m}$ are located at nearby regions (i.e., area $s_T$ and area $s_R$, respectively) of $\bar{\bm{t}}_{n}$ and $\bar{\bm{r}}_{m}$, respectively.
The distance between $\bm{t}_{n}$ and $\bm{r}_{m}$ can be alternatively determined as $d_{mn} = \left \| {\bar{\bm{r}}_{m} - \bar{\bm{t}}_{n}} + \Delta \bm{r}_{m} - \Delta \bm{t}_{n} \right \|_{2}$.
Since both the transmit and receive antenna elements are infinitesimal compared to the distance between them, $d_{mn}$ can therefore be approximated by the distance between the area centers, namely, $d_{mn} \approx \bar{d}_{mn} \triangleq \left \| \bar{\bm{d}}_{mn} \right \|_{2} = \left \| {\bar{\bm{r}}_{m} - \bar{\bm{t}}_{n}} \right \|_{2}$. \footnote{The feasibility discussion of the approximation for parallel surface placements is presented in \cite{Wei2022Tri}, which can be a preliminary support of our approximation for arbitrary surface placements.}  We use this approximation to the amplitude term of tensor Green's function and keep its phase term unchanged, yielding the approximated tensor Green's function 
\begin{align}
	\label{eq:GreenFunc-2}
	{\bm{G}}\left( {{{\bm{r}}_m},{{\bm{t}}_n}} \right)
	\approx {\bm{A}_{mn}} \cdot e^{i{k_0} \left \| {\bar{\bm{r}}_{m} - \bar{\bm{t}}_{n}} + \Delta \bm{r}_{m} - \Delta \bm{t}_{n} \right \|_{2}}, 
\end{align}
where the approximated amplitude term is given by 
\begin{align}
	\nonumber
	\bm{A}_{mn}
	&= \frac{-i}{{4\pi \bar{d}_{mn}}} \left[ \left( {1 + \frac{i}{{{k_0} \bar{d}_{mn}}} - \frac{1}{{k_0^2{\bar{d}_{mn}^2}}}} \right) {{\bm{I}}_3} \right. \\
	\nonumber
	&+ \left. \left( {\frac{3}{{k_0^2{\bar{d}_{mn}^2}}} - \frac{{3i}}{{{k_0} \bar{d}_{mn}}} - 1} \right) \frac{ \bar{\bm{d}}_{mn} \bar{\bm{d}}_{mn}^{T} } {{{\bar{d}_{mn}^2}}} \right]. 
\end{align}

\begin{figure}[t!]
	\centering
	\includegraphics[height=3.8cm, width=7.2cm]{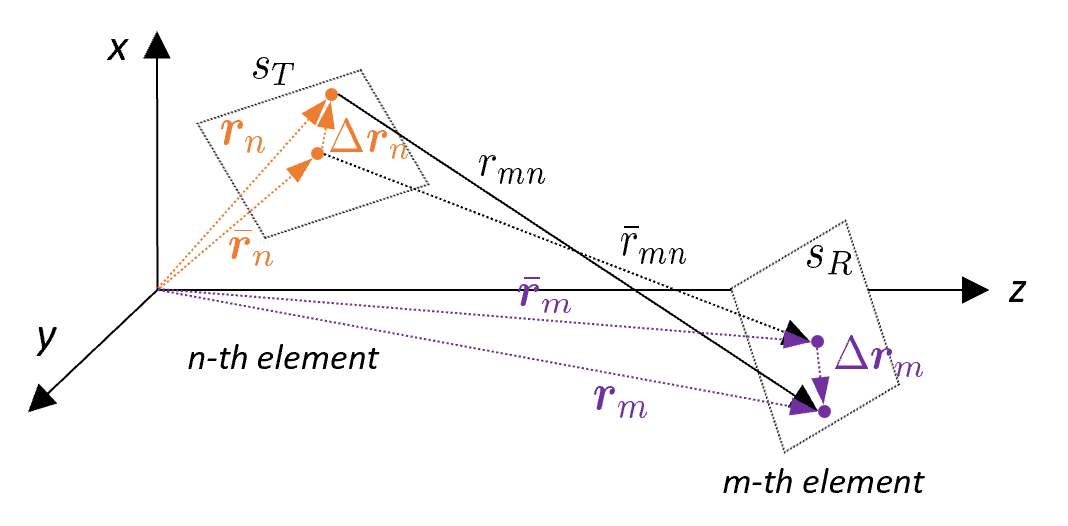}
	\caption{Representation of an arbitrary point, located on the $n$-th TX element or the $m$-th RX element, by the element center and nearby regions.}
	\label{fig:ElementRepresentation}
	\vspace{-0.5em}
\end{figure}

Plugging \eqref{eq:GreenFunc-2} into \eqref{eq:MP2MP-Channel}, one can directly observe that the approximated amplitude ${\bm{A}_{mn}}$ of the tensor Green's function can be extracted outside the integrals, thereby resulting in
\begin{align}
	\nonumber
	\bm{H}_{mn} &\approx \frac{\eta}{2 \lambda} \bm{A}_{mn} \int_{s_R} \int_{s_T} e^{i{k_0} \left \| {\bar{\bm{r}}_{m} - \bar{\bm{t}}_{n}} + \Delta \bm{r}_{m} - \Delta \bm{t}_{n} \right \|_{2}} \\
	\label{eq:MP2MP-Channel-1}
	&\qquad \qquad \cdot {\rm{d}}{\Delta {\bm{t}}_n} {\rm{d}}{\Delta {\bm{r}}_m}.
\end{align}
In order to further proceed with the integrals, we resort to the Taylor series expansion, which states that a scalar function $f\left( {{{\bm{x}}_0} + {\bm{x}}} \right)$ can be expanded as $f\left( {{{\bm{x}}_0} + {\bm{x}}} \right) = f\left( {{{\bm{x}}_0}} \right) + \nabla f{{\left( {{{\bm{x}}_0}} \right)}^T}{\bm{x}} + \frac{1}{2}{{\bm{x}}^T}{\nabla ^2}f\left( {{{\bm{x}}_0}} \right){\bm{x}} + o\left( {\left\| {\bm{x}} \right\|_2^2} \right)$, where $o\left( \cdot \right)$ denotes the high-order terms. As such, we can expand the distance $d_{mn}$ accordingly and obtain its approximation as 
\begin{align}
	\nonumber
	d_{mn} 
	&= \bar{d}_{mn} 
	+ \frac{\bar{\bm{d}}_{mn}^{T}}{\bar{d}_{mn}}\left( {\Delta {{\bm{r}}_m} - \Delta {{\bm{t}}_n}} \right) \\
	\nonumber
	&\qquad \quad + \frac{1}{2}{\left( {\Delta {{\bm{r}}_m} - \Delta {{\bm{t}}_n}} \right)^T} \mathcal{H} \left( {\Delta {{\bm{r}}_m} - \Delta {{\bm{t}}_n}} \right) \\
	\nonumber
	&\qquad \quad + o \left( \left\| {\Delta {{\bm{r}}_m} - \Delta {{\bm{t}}_n}} \right\|_2^2 \right) \\
	\label{eq:TaylorExpansion}
	&\approx \bar{d}_{mn} 
	+ \frac{\bar{\bm{d}}_{mn}^{T}}{\bar{d}_{mn}}\left( {\Delta {{\bm{r}}_m} - \Delta {{\bm{t}}_n}} \right), 
\end{align}
where $\mathcal{H} = {{\nabla ^2}{{\left\| {{{{\bm{\bar r}}}_m} - {{{\bm{\bar t}}}_n}} \right\|}_2}}$ denotes the Hessian matrix. 
The approximation in \eqref{eq:TaylorExpansion} holds by eliminating the second-order and high-order terms of the expansion, which is reasonable as they can be neglected for infinitesimal antenna elements. Substituting \eqref{eq:TaylorExpansion} back into \eqref{eq:MP2MP-Channel-1}, we get
\begin{align}
	\nonumber
	\bm{H}_{mn} 
	&\approx \frac{\eta}{2 \lambda}  \bm{A}_{mn} e^{ i {k_0} \bar{d}_{mn} } \\
	\nonumber 
	&\times \int\nolimits_{s_R} e^{ i {k_0}  \frac{\bar{\bm{d}}_{mn}^{T}}{\bar{d}_{mn}} \Delta {{\bm{r}}_m} } {\rm{d}}{\Delta {\bm{r}}_m} 
	\int\nolimits_{s_T} e^{ - i {k_0}  \frac{\bar{\bm{d}}_{mn}^{T}}{\bar{d}_{mn}} \Delta {{\bm{t}}_n} } {\rm{d}}{\Delta {\bm{t}}_n} \\
	\label{eq:MP2MP-Channel-2}
	&\triangleq \frac{\eta}{2 \lambda} \bm{A}_{mn} e^{ i {k_0} \bar{d}_{mn} }  I_{R} I_{T},
\end{align}
where the double integrals are decomposed into two individual integrals, $I_T$ and $I_R$. 
It is difficult to further solve the integrals $I_{R}$ and $I_{T}$ because the integral region and the integral variable have a relationship involving the arbitrariness of TX (RX) surface placements that need to be clearly characterized. To this end, in the subsequent sections, we present an effective approach for characterizing the arbitrariness of surface placements.

\subsection{Arbitrariness Characterization of Surface Placements}

\begin{figure}[t!]
	\centering
	\includegraphics[height=4.0cm, width=8.6cm]{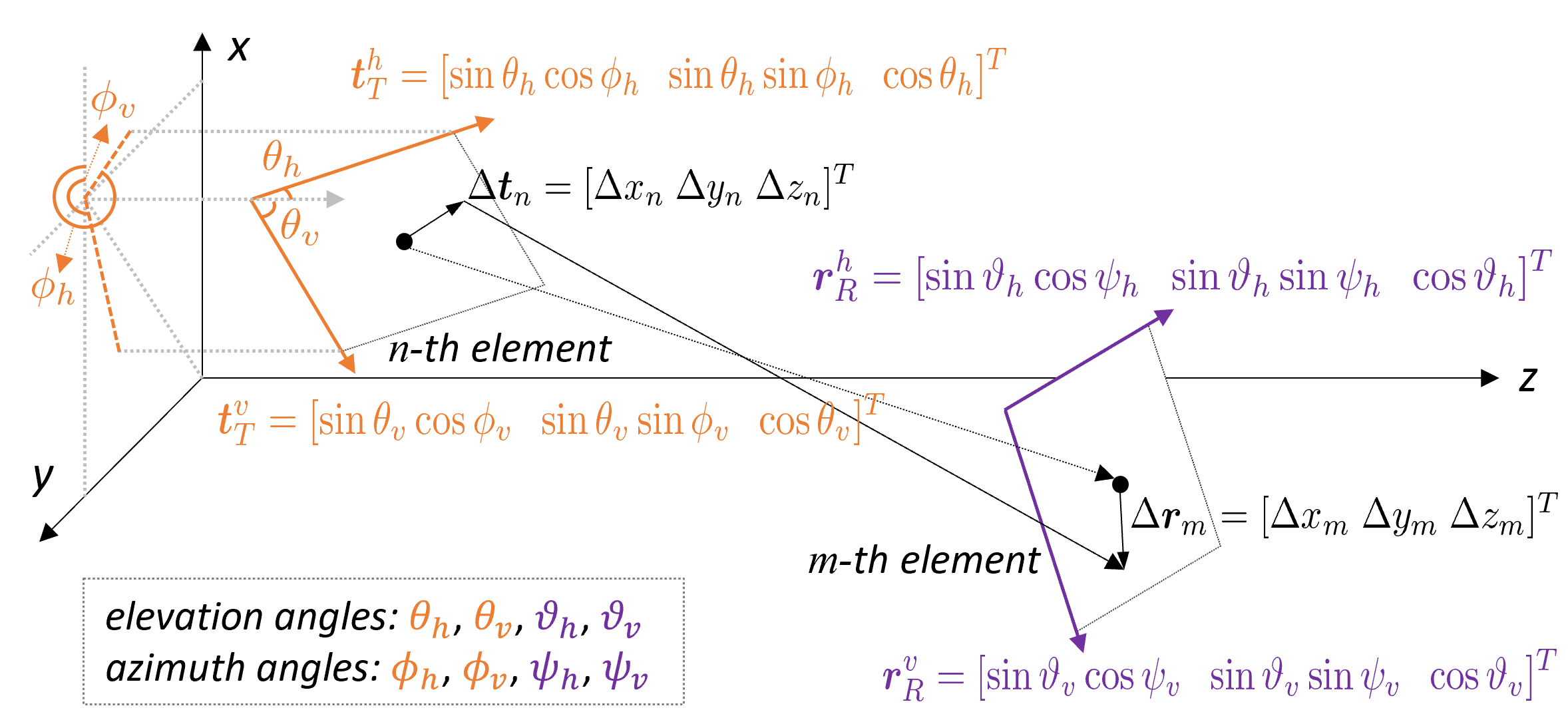}
	\caption{The proposed approach for characterizing the arbitrariness of surface placements.}
	\label{fig:CartesianCoordinateSystem}
\end{figure}

It is well-known that a surface located in the Cartesian coordinate system can be fully characterized by the center coordinates and two associated unit vectors. 

Without loss of generality, we apply the horizontal unit vector and the vertical unit vector of the surface as those two unit vectors, and define them for the TX surface as
\begin{align}
	\label{eq:TESurfaceHUV}
	{\bm{t}}_T^h &= [\sin \theta_{h} \cos \phi_{h} \;\; \sin \theta_{h} \sin \phi_{h} \;\; \cos \theta_{h}]^{T}, \\
	\label{eq:TESurfaceVUV}
	{\bm{t}}_T^v &= [\sin \theta_{v} \cos \phi_{v} \;\; \sin \theta_{v} \sin \phi_{v} \;\; \cos \theta_{v}]^{T},
\end{align}
where $\theta_h$ and $\theta_v$ are the elevation angles of the horizontal direction and the vertical direction of the TX surface, respectively, i.e., the angle between the $z$-axis and the horizontal/vertical direction of the TX surface as shown in Fig. \ref{fig:CartesianCoordinateSystem};  $\phi_h$ and $\phi_v$ are the azimuth angles of the horizontal direction and the vertical direction of the TX surface, respectively, namely, the angle between $x$-axis and the $xy$-plane projection of the horizontal/vertical direction. 
Likewise, we define the horizontal unit vector and the vertical unit vector for the RX surface as
\begin{align}
	\label{eq:RESurfaceHUV}
	{\bm{r}}_R^h &= [\sin \vartheta_{h} \cos \psi_{h} \;\; \sin \vartheta_{h} \sin \psi_{h} \;\; \cos \vartheta_{h}]^{T}, \\
	\label{eq:RESurfaceVUV}
	{\bm{r}}_R^v &= [\sin \vartheta_{v} \cos \psi_{v} \;\; \sin \vartheta_{v} \sin \psi_{v} \;\; \cos \vartheta_{v}]^{T},
\end{align}
where $\vartheta_h$, $\vartheta_v$, $\psi_h$, and $\psi_v$ are the elevation angles and the azimuth angles defined for the RX surface.
These defined angles are capable of controlling the orientation of antenna surfaces. 
It is worth noting that they also control the surface shape (and of course the element shape). 
The following lemma provides the condition that guarantees the surface shape (element shape) to be rectangle (or square), otherwise, the surface shape (element shape) becomes diamond.
\begin{lemma}
	\label{LemmaSurfaceShape}
	For the TX/RX surfaces with its horizontal unit vector and vertical unit vector given by \eqref{eq:TESurfaceHUV}, \eqref{eq:RESurfaceHUV} and \eqref{eq:TESurfaceVUV}, \eqref{eq:RESurfaceVUV}, the TX/RX surfaces and their elements are guaranteed to be rectangle (or square) if and only if 
	\begin{flalign}
		\left\{ \begin{aligned}
			&{{\theta _h} = {90^o} \; {\rm{or}} \; {\theta _v} = {90^o}},&  &{{\rm{for \; TX}}}, \\
			&{{\vartheta _h} = {90^o} \; {\rm{or}}  \; {\vartheta _v} = {90^o}},& &{{\rm{for \; RX}}}, \\
		\end{aligned}  \right.
	\end{flalign}
\end{lemma}
\begin{proof}
	See appendix \ref{LemmaSurfaceShapeProof}.
\end{proof}

More importantly, the proposed arbitrariness characterization approach allows us to express the $z$ orientation as a function of the $x$ and $y$ orientations. As a consequence, the integrals in $I_{T}$ and $I_{R}$ can be feasibly solved in closed form. 
The following lemma unveils their connections, which will be utilized in the following near-field LoS channel modeling.

\begin{lemma}
	\label{LemmaZbyXY}
	Suppose that TX and RX antenna elements with area $s_{T}$ and area $s_{R}$ are located in the Cartesian coordinate system, with their horizontal and vertical unit vectors given by \eqref{eq:TESurfaceHUV} \eqref{eq:TESurfaceVUV} and \eqref{eq:RESurfaceHUV} \eqref{eq:RESurfaceVUV}, respectively. For any arbitrary point $\Delta {\bm{t}}_n \in s_T$ and any arbitrary point $\Delta {\bm{r}}_m \in s_R$, their $\Delta z$ components can be jointly represented by the $\Delta x$ component and the $\Delta y$ component, respectively, as
	\begin{align}
		\nonumber
		\Delta z_n 
		&=\frac{\cot \theta_v \sin \phi_h -\cot \theta_h \sin \phi_v}{\sin \left(\phi_h-\phi_v\right)} \Delta x_n \\
		\label{eq:DeltaZn}
		&+ \frac{\cot \theta_v \cos \phi_h -\cot \theta_h \cos \phi_v}{\sin \left(\phi_h-\phi_v\right)} \Delta y_n, \\
		\nonumber
		\Delta {z_m} &= \frac{ \cot \vartheta_v \sin \psi_{h} - \cot \vartheta_h \sin \psi_v }{\sin (\psi_{h}-\psi_v)} \Delta x_m \\
		\label{eq:DeltaZm}
		&+ \frac{ \cot \vartheta_v \cos \psi_h - \cot \vartheta_h \cos \psi_{v} }{\sin (\psi_{h}-\psi_v)} \Delta y_m.
	\end{align}
\end{lemma}
\begin{proof}
	See appendix \ref{LemmaZbyXYProof}.
\end{proof}

It is remarkable that Lemma \ref{LemmaZbyXY} provides the results without restricting the surface shape to be diamond or rectangle (square). If the rectangle (square) shape is chosen, $\theta_{h}$ or $\theta_{v}$ for TX and $\vartheta_{h}$ or $\vartheta_{v}$ for RX should be $90^o$. Meanwhile, $|\phi_h-\phi_v| = 90^o$ and $|\psi_{h}-\psi_v| = 90^o$ hold as well. As a result, $\Delta z_n$ and $\Delta {z_m}$ reduce to more compact forms, where we omit for brevity.

\subsection{Coordinate-Dependent Channel Model (CD-CM)}

It is worth noting that the integral regions of $I_{T}$ and $I_{R}$ are restricted to their corresponding element surface areas, respectively, which does not correspond to the volume property of ${\rm{d}}{\Delta {\bm{t}}_n}$ and ${\rm{d}}{\Delta {\bm{r}}_m}$ in a point to point manner. 
Lemma \ref{LemmaZbyXY} allows us to tackle the issue, facilitating the derivations of $I_{T}$ and $I_{R}$. One can obtain $\bm{H}_{mn}$ by plugging the results in Lemma \ref{LemmaZbyXY} into \eqref{eq:MP2MP-Channel-2}. The following theorem demonstrates our derived result.
\begin{theorem}
	\label{TheoremChannelModel}
	For communications occurred between an $M$-element TX surface and an $N$-element RX surface, both located in the Cartesian coordinate system with arbitrary placements, and each surface comprises antenna elements with an equal element area, $s_{T}$ and $s_{R}$ for TX and RX, respectively, the near-field LoS channel for the $mn$-pair is established as
	\begin{align}
		\label{eq:MP2MP-Channel-3}
		\bm{H}_{mn}^{\text{CD-CM}} \approx \frac{\eta}{2 \lambda} \cdot  l_R^h l_R^v \cdot l_T^h l_T^v \cdot \bm{A}_{mn} \cdot e^{ i {k_0} \bar{d}_{mn} } \cdot \varrho_{mn},
	\end{align}
	where the factor $\varrho_{mn}$ is given by
	\begin{align}
		\nonumber
		&\varrho_{mn} = {\rm{sinc}}\left[ {\frac{{\pi l_T^h}}{\lambda}  \frac{{{{\bar x}_{mn}} + {{\bar z}_{mn}} \frac{\sin \phi_{h} \cot \theta_v - \sin \phi_v \cot \theta_h }{\sin (\phi_{h}-\phi_v)} }}{{{{\bar d}_{mn}}}}} \right] \\
		\nonumber 
		&\times {\rm{sinc}}\left[ {\frac{{\pi l_T^v}}{\lambda}  \frac{{{{\bar y}_{mn}} + {{\bar z}_{mn}} \frac{ \cos \phi_h \cot \theta_v - \cos \phi_{v} \cot \theta_h }{\sin (\phi_{h}-\phi_v)} }}{{{{\bar d}_{mn}}}}} \right] \\
		\nonumber
		&\times {\rm{sinc}}\left[ {\frac{{\pi l_R^h}}{\lambda}  \frac{{{{\bar x}_{mn}} + {{\bar z}_{mn}} \frac{\sin \psi_{h} \cot \vartheta_v - \sin \psi_v \cot \vartheta_h }{\sin (\psi_{h}-\psi_v)} }}{{{{\bar d}_{mn}}}}} \right] \\
		\label{eq:eta_mn}
		&\times {\rm{sinc}}\left[ {\frac{{\pi l_R^v}}{\lambda}  \frac{{{{\bar y}_{mn}} + {{\bar z}_{mn}} \frac{ \cos \psi_h \cot \vartheta_v - \cos \psi_{v} \cot \vartheta_h }{\sin (\psi_{h}-\psi_v)} }}{{{{\bar d}_{mn}}}}} \right].
	\end{align}
\end{theorem}
\begin{proof}
	See appendix \ref{TheoremChannelModelProof}.
\end{proof}

Theorem \ref{TheoremChannelModel} provides a generalized channel model for both diamond and rectangle (square) surfaces. Note that for the diamond surface, $l_T^h l_T^v$ and $l_R^h l_R^v$ cannot be further represented by area $s_{T}$ and $s_{R}$. However, this can be realized for rectangle (square) surfaces, given by 
\begin{align}
	\label{eq:MP2MP-Channel-3}
	\bm{H}_{mn}^{\text{CD-CM}} \approx \frac{\eta}{2 \lambda} \cdot  s_{R} s_{T} \cdot \bm{A}_{mn} \cdot e^{ i {k_0} \bar{d}_{mn} } \cdot \varrho_{mn},
\end{align}
where $s_{T} = l_T^h l_T^v$, $s_{R} = l_R^h l_R^v$, and the factor $\varrho_{mn}$ becomes a reduced expression accordingly in cases when ($\theta_{h} = 90^o$ and $\vartheta_{h} = 90^o$) or ($\theta_{h} = 90^o$ and $\vartheta_{v} = 90^o$) or ($\theta_{v} = 90^o$ and $\vartheta_{h} = 90^o$) or ($\theta_{v} = 90^o$ and $\vartheta_{v} = 90^o$), which we omit for brevity.

\textit{Remark:} Compared with the INT-CM \eqref{eq:MP2MP-Channel}, our approximated channel model \eqref{eq:MP2MP-Channel-3} eliminates the integral operators and provides an explicit and computationally-efficient surrogate. This approximated model shows that the $mn$-pair channel can be depicted by the distance between their area centers, their surface areas, as well as the angles of their surface placements. We also point out that this model relies on the absolute location coordinates of the antenna element center in deriving $\bar{x}_{mn}$, $\bar{y}_{mn}$ and $\bar{z}_{mn}$. Moreover, it is remarkable that the channel model in \eqref{eq:MP2MP-Channel-3} generalizes the conventional parallel placement channel model established in \cite{Wei2022Tri}. Specifically, when the TX and RX surfaces are parallel to each other, and are also placed parallel to the $xy$-plane, the elevation angles $\theta_{h}$, $\theta_{v}$, $\vartheta_{h}$, and $\vartheta_{v}$ will be $90^o$, and the terms including $\bar{z}_{mn}$ will reduce to zero, resulting in the channel model derived in \cite{Wei2022Tri}.

\subsection{Coordinate-Independent Channel Model (CI-CM)}

The CD-CM shows an excellent approximation to the INT-CM (shown in numerical evaluations). It is however sophisticated in its expression for further applications, and also requires the concrete coordinates of each antenna element, which is possibly impractical in many cases. Instead, we may aware of the relative distance and direction between TX and RX antenna elements. As such, it is essential to develop a concise coordinate-independent but effective alternative, presented by the following corollary.

\begin{corollary}
	\label{CorollaryForCICM}
	Consider communications between a rectangle $M$-element TX surface with an area of $s_{T}$ and a rectangle $N$-element RX surface with an area of $s_{R}$. The near-field LoS channel between each $mn$-pair is simplified to a coordinate-independent alternative, given by
	\begin{align}
		\label{eq:MP2MP-Channel-4}
		\bm{H}_{mn}^{\text{CI-CM}} \approx \frac{\eta}{2 \lambda}  s_R s_T  \bm{A}_{mn} e^{ i {k_0} \bar{d}_{mn} }  
		= \frac{\eta}{2 \lambda} s_R s_T {\bm{G}}\left( \bar{\bm{d}}_{mn} \right).
	\end{align}
\end{corollary}
\begin{proof}
	To relieve the dependency on the absolute location coordinates, 
	we resort to the Taylor series expansion to the ${\rm{sinc}}(x)$ function, namely, ${\rm{sinc}}{(x)}=1-\frac{x^{2}}{3 !}+\frac{x^{4}}{5 !}-\frac{x^{6}}{7 !}+\ldots = \sum_{q=0}^{\infty} \frac{(-1)^{q} x^{2 q}}{(2 q+1) !}$. We notice that, for infinitesimal antenna elements, the high-order terms in the expansion can be neglected, especially when several $\rm{sinc}$ functions multiplying with each other. We thus get 
	\begin{align}
		\nonumber
		\varrho_{mn} \approx& 1 - \nu  \left( \frac{{ l_T^h }}{\lambda } \right)^{2}  \left( { \frac{{\bar x}_{mn}}{{\bar z}_{mn}} + \frac{\sin \phi_{h} \cot \theta_v - \sin \phi_v \cot \theta_h }{\sin (\phi_{h}-\phi_v)} } \right)^{2} \\
		\nonumber
		&- \nu  \left( \frac{{ l_T^v }}{\lambda} \right)^{2}  \left( { \frac{{\bar y}_{mn}}{{\bar z}_{mn}} + \frac{ \cos \phi_h \cot \theta_v - \cos \phi_{v} \cot \theta_h }{\sin (\phi_{h}-\phi_v)} } \right)^{2} \\
		\nonumber
		&- \nu \left( \frac{{ l_R^h }}{\lambda } \right)^{2}  \left( { \frac{{\bar x}_{mn}}{{\bar z}_{mn}} + \frac{\sin \psi_{h} \cot \vartheta_v - \sin \psi_v \cot \vartheta_h }{\sin (\psi_{h}-\psi_v)} } \right)^{2} \\
		\nonumber
		&- \nu \left( \frac{{ l_R^v }}{\lambda } \right)^{2}  \left( { \frac{{\bar y}_{mn}}{{\bar z}_{mn}} + \frac{ \cos \psi_h \cot \vartheta_v - \cos \psi_{v} \cot \vartheta_h }{\sin (\psi_{h}-\psi_v)} } \right)^{2} \\
		\nonumber
		\approx& 1,
	\end{align}
	where $\nu = \frac{ \pi^2 }{3 !} \times \left( \frac{{\bar z}_{mn}}{{\bar d}_{mn}} \right)^{2}$. Substituting the result to \eqref{eq:MP2MP-Channel-3}, we obtain the expected result, which completes the proof.
\end{proof}

To explicitly demonstrate the reasonability, we provide the following numerical example, in which we assume ${\bar x}_{mn} = {\bar y}_{mn} = 0$ for brevity, such that ${\bar d}_{mn} = {\bar z}_{mn}$. We further set $\phi_{h} = \psi_{h} = 90^{o}$, $\phi_{v} = \psi_{v} = 0^{o}$, $\theta_{h} = \theta_{v} = \vartheta_{h} = 90^{o}$, $\vartheta_{v} = 60^{o}$, implying that the angle between the TX and RX surfaces is $30^{o}$. Therefore, $\frac{\sin \phi_{h} \cot \theta_v - \sin \phi_v \cot \theta_h}{\sin (\phi_{h}-\phi_v)} = 0$, $\frac{\sin \psi_{h} \cot \vartheta_v - \sin \psi_v \cot \vartheta_h}{\sin (\psi_{h}-\psi_v)} = \frac{\sqrt{3}}{3}$ and $\frac{ \cos \phi_h \cot \theta_v - \cos \phi_{v} \cot \theta_h}{\sin (\phi_{h}-\phi_v)} = \frac{\cos \psi_h \cot \vartheta_v - \cos \psi_{v} \cot \vartheta_h}{\sin (\psi_{h}-\psi_v)} = 0$. Since both TX and RX surfaces are nearly continuous, we make a moderate assumption that $l_{T}^{h} = l_{T}^{v} = l_{R}^{h} = l_{R}^{v} = \frac{\lambda}{5}$. Under these setups, we can calculate $\varrho_{mn} = 0.978$, which approximates to $1$ with a negligible error.

It can be observed from \eqref{eq:MP2MP-Channel-4} that the near-field LoS channel between $mn$-pair can be depicted by their surface areas, and the tensor Green's function in terms of their center distance $\bar{d}_{mn}$ and the direction (unit vector) $\bar{\bm{d}}_{mn}/\bar{d}_{mn}$. Note that $\bar{d}_{mn}$ and $\bar{\bm{d}}_{mn}/\bar{d}_{mn}$ are irrelevant to the absolute coordinates of each antenna element, which can be acquired by distance and angle measurements.

\section{EM-domain Capacity Limit}
\label{SectionCL}

Based upon our established near-field LoS channel model, it is capable of analyzing the capacity limit of the H-MIMO system from the EM-domain perspective. To realize this goal, we first build an effective analytical framework, where the EM-domain transmit and receive patterns (see Definition \ref{Definition} below) are defined and introduced into the communication model, and an effective channel decomposition is then built for facilitating the analysis. We then unveil the fundamental capacity limit based upon this analytical framework in a rigorous derivation.

\subsection{Transmit and Receive Patterns}
For the purpose of rigorously formulating the communication model with transmit and receive patterns, and performing analysis on the fundamental limit, we first present the following definitions of the transmit and receive patterns.

\begin{definition}
	\label{Definition}
	The transmit or receive patterns of an H-MIMO communication system indicate a set of completed orthogonal bases $\left\{ \bm{\tau}_{1}(\bm{r}), \bm{\tau}_{2}(\bm{r}), \cdots, \bm{\tau}_{P}(\bm{r}) \right\}$ with $\bm{\tau}_{p}(\bm{r}) = [\tau_{p}^{x}(\bm{r}), \tau_{p}^{y}(\bm{r}), \tau_{p}^{z}(\bm{r})]^{T}$ defined on the surface area $\bm{r} \in S$, capable of representing any arbitrary function defined on this surface area \cite{Piestun2000Electromagnetic, Miller2000Communicating}, such as the current distribution $\bm{j}$ or the electric field $\bm{e}$, namely, 
	\begin{align}
		\bm{o}(\bm{r}) = \sum_{p=1}^{P} \omega_{p} \bm{\tau}_{p}(\bm{r}), \;\; \bm{o} \in \{ \bm{j}, \bm{e} \}, 
	\end{align}
	where $\omega_{p}$ corresponds to the weight of each basis vector.
	Particularly, orthonormality holds
	\begin{align}
		\label{eq:Orthogonality}
		\int_{A} \bm{\tau}_{p}^{H}(\bm{r}) \bm{\tau}_{q}(\bm{r}) \mathrm{d} \bm{r} = \delta_{pq},
	\end{align}
	where $\delta_{p q}$ is the Kronecker delta function, given by $\delta_{p q} = 1$ when $p = q$ while $\delta_{p q} = 0$ otherwise. 
\end{definition}

Following this definition, we denote the transmit patterns by $\mathcal{S}_{u} = \left\{ \bm{u}_{1}(\bar{\bm{t}}_n), \bm{u}_{2}(\bar{\bm{t}}_n), \cdots, \bm{u}_{P}(\bar{\bm{t}}_n) \right\}$ on the area of the TX surface $A_{T}$, capable of generating the current distribution intended to be radiated. Likewise, we define the receive patterns by $\mathcal{S}_{v} = \left\{ \bm{v}_{1}(\bar{\bm{r}}_m), \bm{v}_{2}(\bar{\bm{r}}_m), \cdots, \bm{v}_{Q}(\bar{\bm{r}}_m) \right\}$ on the area of the RX surface $A_{R}$ to represent the measured electric field. Consequently, we get
\begin{align}
	\label{eq:TransmitPattern}
	\bm{j}_n(\bar{\bm{t}}_n) &= \sum_{p=1}^{P} a_{p} \bm{u}_{p}(\bar{\bm{t}}_n), \\
	\label{eq:ReceivePattern}
	{\bm{e}} \left( {{\bar{\bm{r}}_m}} \right) 
	&= \sum_{q=1}^{Q} c_{q} \bm{v}_{q}(\bar{\bm{r}}_m),
\end{align}
where $a_{p}$ and $c_{q}$ are the corresponding scalar weights. It is noteworthy that $a_{p}$ and $c_{q}$ can be the transmit symbol and the estimated receive symbol, respectively. In this regards, $P$ equals $Q$, and both indicate the number of symbol streams. 
Based upon the representation by patterns, we establish the communication model of H-MIMO systems between the transmit and receive symbols. To assist the system modeling, we define the transmit and receive pattern matrices as 
\begin{align}
	\label{eq:Tmatrix}
	\bm{T} &= \left[ {\begin{array}{*{20}{c}}
			{{{\bm{u}}_{1}}\left( {{{\bar{\bm{t}}}_1}} \right)}&{{{\bm{u}}_{2}}\left( {{{\bar{\bm{t}}}_1}} \right)}& \cdots &{{{\bm{u}}_{P}}\left( {{{\bar{\bm{t}}}_1}} \right)}\\
			{{{\bm{u}}_{1}}\left( {{{\bar{\bm{t}}}_2}} \right)}&{{{\bm{u}}_{2}}\left( {{{\bar{\bm{t}}}_2}} \right)}& \cdots &{{{\bm{u}}_{P}}\left( {{{\bar{\bm{t}}}_2}} \right)}\\
			\vdots & \vdots & \ddots & \vdots \\
			{{{\bm{u}}_{1}}\left( {{{\bar{\bm{t}}}_N}} \right)}&{{{\bm{u}}_{2}}\left( {{{\bar{\bm{t}}}_N}} \right)}& \cdots &{{{\bm{u}}_{P}}\left( {{{\bar{\bm{t}}}_N}} \right)}
	\end{array}} \right], 
	\\
	\label{eq:Rmatrix}
	\bm{R} &= {\left[ {\begin{array}{*{20}{c}}
				{{{\bm{v}}_{1}}\left( {{{\bar{\bm{r}}}_1}} \right)}&{{{\bm{v}}_{2}}\left( {{{\bar{\bm{r}}}_1}} \right)}& \cdots &{{{\bm{v}}_{Q}}\left( {{{\bar{\bm{r}}}_1}} \right)}\\
				{{{\bm{v}}_{1}}\left( {{{\bar{\bm{r}}}_2}} \right)}&{{{\bm{v}}_{2}}\left( {{{\bar{\bm{r}}}_2}} \right)}& \cdots &{{{\bm{v}}_{Q}}\left( {{{\bar{\bm{r}}}_2}} \right)}\\
				\vdots & \vdots & \ddots & \vdots \\
				{{{\bm{v}}_{1}}\left( {{{\bar{\bm{r}}}_M}} \right)}&{{{\bm{v}}_{2}}\left( {{{\bar{\bm{r}}}_M}} \right)}& \cdots &{{{\bm{v}}_{Q}}\left( {{{\bar{\bm{r}}}_M}} \right)}
		\end{array}} \right] }, 
\end{align}
based on which we obtain 
\begin{align}
	\label{eq:PatternsForJ}
	\bar{\bm{j}} = \bm{T} \bm{a}, \\
	\label{eq:PatternsForE}
	\bar{\bm{e}} = \bm{R} \bm{c},
\end{align}
where $\bm{a} = [{{a_1}}, {{a_2}}, \cdots, {{a_P}}]^{T}$ and ${\bm{c}} = [c_{1}, c_{2}, \cdots, c_{Q}]^{T}$ are the transmit and receive symbol vectors, respectively; $\bar{\bm{j}} \triangleq [{\bm{j}} \left( {{\bar{\bm{t}}_1}} \right), {\bm{j}} \left( {{\bar{\bm{t}}_2}} \right), \cdots, {\bm{j}} \left( {{\bar{\bm{t}}_N}} \right)]^{T}$ represents the vector consisting of current distributions at all $N$ area centers of TX antenna elements; $\bar{\bm{e}} \triangleq [{\bm{e}} \left( {{\bar{\bm{r}}_1}} \right), {\bm{e}} \left( {{\bar{\bm{r}}_2}} \right), \cdots, {\bm{e}} \left( {{\bar{\bm{r}}_M}} \right)]^{T}$ denotes the vector consisting of electric fields at all $M$ area centers of RX antenna elements.


The transmit and receive patterns have their own properties, relying on the following lemma.
\begin{lemma}
	\label{LemmaIntegralApprox}
	Suppose that a surface area $A$, composing $K$ uniformly divided small element areas $S$. Let $\{ \bar{\bm{r}}_{1}, \bar{\bm{r}}_{2}, \cdots, \bar{\bm{r}}_{K} \}$ be the element centers. For a limited area surface with infinitesimal antenna elements, the integral of any two differentiable functions, defined on this surface area, can be well depicted using $K$, $S$ and the element centers, namely, 
	\begin{align}
		\int_{A} \bm{f}^{H}(\bm{r}) \bm{g}(\bm{r}) \mathrm{d} \bm{r} 
		\approx S \sum_{k=1}^{K} \bm{f}^{H}(\bar{\bm{r}}_{k}) \bm{g}(\bar{\bm{r}}_{k}),
	\end{align}
	where $\bm{f}(\bm{r})$ and $\bm{g}(\bm{r})$ are those two differentiable functions defined on $A$.
\end{lemma}
\begin{proof}
	See appendix \ref{LemmaIntegralApproxProof}.
\end{proof}

Lemma \ref{LemmaIntegralApprox} allows us to unveil the properties for the transmit and receive patterns, as well as the transmit and receive pattern matrices. It is stated by the following corollary.

\begin{corollary}
	\label{CorollaryForPatterns}
	For TX and RX surfaces with limited area and infinitesimal antenna elements, the transmit patterns $\mathcal{S}_{u} = \left\{ \bm{u}_{1}(\bar{\bm{t}}_n), \bm{u}_{2}(\bar{\bm{t}}_n), \cdots, \bm{u}_{P}(\bar{\bm{t}}_n) \right\}$ defined on the area of TX surface $A_{T} = N s_{T}$, and the receive patterns $\mathcal{S}_{v} = \left\{ \bm{v}_{1}(\bar{\bm{r}}_m), \bm{v}_{2}(\bar{\bm{r}}_m), \cdots, \bm{v}_{Q}(\bar{\bm{r}}_m) \right\}$ defined on the area of RX surface $A_{R} = M s_{R}$, nearly obey orthonormality as follows
	\begin{align}
		\label{eq:PatternMatProperty}
		\bm{T}^{H} \bm{T} \approx \frac{\bm{I}_{P}}{s_{T}}, 
		\;\;\;\;
		\bm{R}^{H} \bm{R} \approx \frac{\bm{I}_{Q}}{s_{R}}. 
	\end{align}
\end{corollary}
\begin{proof}
	Based on Lemma \ref{LemmaIntegralApprox}, we approximate the pattern orthonormality \eqref{eq:Orthogonality} as 
	\begin{align}
		\nonumber
		\delta_{pq} = \int_{A} \bm{\tau}_{p}^{H}(\bm{r}) \bm{\tau}_{q}(\bm{r}) \mathrm{d} \bm{r} 
		\approx S \sum_{k=1}^{K}  \bm{\tau}_{p}^{H}(\bar{\bm{r}}_{k}) \bm{\tau}_{q}(\bar{\bm{r}}_{k}), 
	\end{align}
	where $A \in \{ A_{T}, A_{R} \}$, $\bm{\tau} \in \{ \bm{u}, \bm{v} \}$, $S \in \{ s_{T}, s_{R} \}$, and $K \in \{ N, M \}$. By dividing $S$ on both sides of the approximation, we derive \eqref{eq:PatternMatProperty}, which completes the proof. 
\end{proof}

Next, we integrate the transmit and receive patterns into the H-MIMO system model \eqref{eq:P2P-HMIMO-1}, where we rely on the established channel model \eqref{eq:MP2MP-Channel-4} for brevity to derive the elements of the channel matrix. 
In addition, we connect $\bar{\bm{j}}$ in \eqref{eq:PatternsForJ} and $\bm{j}$ in \eqref{eq:P2P-HMIMO-1} based on the uniformly distributed assumption of the current distribution, $\bm{j}_n = \bm{j}_n(\bar{\bm{t}}_n)$, such that $\bm{j} = \bar{\bm{j}}$.
Likewise, we bridge the connection between $\bar{\bm{e}}$ in \eqref{eq:PatternsForE} and $\bm{e}$ in \eqref{eq:P2P-HMIMO-1} via $\bm{e}_m = \int_{s_R} {\bm{e}} \left( {{{\bm{r}}_m}} \right) {\rm{d}}{{\bm{r}}_m} \approx s_{R} \cdot {\bm{e}} \left( {{\bar{\bm{r}}_m}} \right)$
for infinitesimal area of the antenna elements, such that $\bm{e} \approx s_{R} \cdot \bar{{\bm{e}}}$.
As a consequence, combining \eqref{eq:P2P-HMIMO-1}, \eqref{eq:PatternsForJ}, \eqref{eq:PatternsForE}, we obtain the communication model using the pattern matrix property \eqref{eq:PatternMatProperty}, yielding 
\begin{align}
	\label{eq:CommunicationModel}
	{\bm{c}} = \frac{\eta}{2 \lambda} \cdot s_R s_T \cdot \bm{R}^{H} \bm{G} \bm{T} \bm{a} + s_{R} \bm{R}^{H} \bm{w}, 
\end{align}
where $\bm{G}$ is specified to satisfy $\bm{H} = \frac{\eta}{2 \lambda} s_R s_T \bm{G}$, and it will be detailed subsequently.

\subsection{Channel Decomposition}
In building the communication model \eqref{eq:CommunicationModel},
$\bm{G}$ is specified as a matrix composed of the tensor Green's function in terms of different center coordinates of TX/RX antenna elements. In particular, it is provided by
\begin{align}
	\label{eq:Gmatrix}
	\bm{G} = \left[ {\begin{array}{*{20}{c}}
			{{\bm{G}}\left( {{{\bar{\bm{r}}}_1},{{\bar{\bm{t}}}_1}} \right)}&{{\bm{G}}\left( {{{\bar{\bm{r}}}_1},{{\bar{\bm{t}}}_2}} \right)}& \cdots &{{\bm{G}}\left( {{{\bar{\bm{r}}}_1},{{\bar{\bm{t}}}_N}} \right)}\\
			{{\bm{G}}\left( {{{\bar{\bm{r}}}_2},{{\bar{\bm{t}}}_1}} \right)}&{{\bm{G}}\left( {{{\bar{\bm{r}}}_2},{{\bar{\bm{t}}}_2}} \right)}& \cdots &{{\bm{G}}\left( {{{\bar{\bm{r}}}_2},{{\bar{\bm{t}}}_N}} \right)}\\
			\vdots & \vdots & \ddots & \vdots \\
			{{\bm{G}}\left( {{{\bar{\bm{r}}}_M},{{\bar{\bm{t}}}_1}} \right)}&{{\bm{G}}\left( {{{\bar{\bm{r}}}_M},{{\bar{\bm{t}}}_2}} \right)}& \cdots &{{\bm{G}}\left( {{{\bar{\bm{r}}}_M},{{\bar{\bm{t}}}_N}} \right)}
	\end{array}} \right]. 
\end{align}
Since the defined transmit patterns $\mathcal{S}_{u}$ and the receive patterns $\mathcal{S}_{v}$ are completed orthogonal bases within area $A_{T}$ and area $A_{R}$, respectively, the function of any arbitrary point belonging to each of the area can be formally represented by the corresponding orthogonal bases. Note that the tensor Green's function ${{\bm{G}}\left( {{{\bar{\bm{r}}}_m},{{\bar{\bm{t}}}_n}} \right)}$ is a function of ${{\bar{\bm{t}}}_n} \in A_{T}$ and ${{\bar{\bm{r}}}_m} \in A_{R}$. Consequently, it can be decomposed using these orthogonal bases. More specifically, for a fixed ${{\bar{\bm{r}}}_m} \in A_{R}$, ${{\bm{G}}\left( {{{\bar{\bm{r}}}_m},{{\bar{\bm{t}}}_n}} \right)}$ is merely a function of ${{\bar{\bm{t}}}_n} \in A_{T}$, which can thus be interpreted by $\mathcal{S}_{u}$ as \cite{Miller2000Communicating}
\begin{align}
	\label{eq:GreenFuncDecomp}
	{{\bm{G}}\left( {{{\bar{\bm{r}}}_m},{{\bar{\bm{t}}}_n}} \right)} = \sum_{p=1}^{P} \bm{h}_{p}(\bar{\bm{r}}_m)
	\bm{u}_{p}^{H}(\bar{\bm{t}}_n), 
\end{align}
where $\bm{h}_{p}(\bar{\bm{r}}_m) 
= \int_{A_{T}} {{\bm{G}}\left( {{{\bar{\bm{r}}}_m},{{\bm{t}}}} \right)} \bm{u}_{p}({\bm{t}}) \mathrm{d} {\bm{t}}$, and it can be approximated via Lemma \ref{LemmaIntegralApprox} by
\begin{align}
	\label{eq:omega_p}
	\bm{h}_{p}(\bar{\bm{r}}_m) 
	\approx \sum_{n=1}^{N} s_{T} \cdot {{\bm{G}}\left( {{{\bar{\bm{r}}}_m},{{\bar{\bm{t}}}_n}} \right)} \bm{u}_{p}(\bar{\bm{t}}_n). 
\end{align}
Accordingly, $\bm{h}_{p}(\bar{\bm{r}}_m)$ is a function of ${{\bar{\bm{r}}}_m} \in A_{R}$, and can be decomposed by $\mathcal{S}_{v}$ as 
\begin{align}
	\label{eq:omega_p_decompose}
	\bm{h}_{p}(\bar{\bm{r}}_m) = \sum_{q=1}^{Q} \gamma_{pq} \bm{v}_{q}(\bar{\bm{r}}_m), 
\end{align}
where $\gamma_{pq} = \int_{A_{R}}  \bm{v}_{q}^{H}({\bm{r}}) \bm{h}_{p}({\bm{r}}) \mathrm{d} {\bm{r}}$, and it can be further obtained using Lemma \ref{LemmaIntegralApprox} as follows
\begin{align}
	\label{eq:gamma_pq} 
	\gamma_{pq} \approx s_{R} s_{T}  \sum_{n=1}^{N} \sum_{m=1}^{M} \bm{v}_{q}^{H}(\bar{\bm{r}}_m) 
	{{\bm{G}}\left( {{{\bar{\bm{r}}}_m},{{\bar{\bm{t}}}_n}} \right)} \bm{u}_{p}(\bar{\bm{t}}_n).
\end{align}
Based upon the decomposition, we substitute \eqref{eq:omega_p_decompose} to \eqref{eq:GreenFuncDecomp} to derive the bilinear decomposition of ${{\bm{G}}\left( {{{\bar{\bm{r}}}_m},{{\bar{\bm{t}}}_n}} \right)}$ with respect to orthogonal bases $\mathcal{S}_{u}$ and $\mathcal{S}_{v}$ by
\begin{align}
	\label{eq:GreenFuncBilinearDecomp}
	{{\bm{G}}\left( {{{\bar{\bm{r}}}_m},{{\bar{\bm{t}}}_n}} \right)} &= \sum_{p=1}^{P} \sum_{q=1}^{Q} \gamma_{pq} \bm{v}_{q}(\bar{\bm{r}}_m)
	\bm{u}_{p}^{H}(\bar{\bm{t}}_n).
\end{align}
By further substituting this bilinear decomposition into \eqref{eq:Gmatrix}, we reformulate the $\bm{G}$ matrix as 
\begin{align}
	\label{eq:GmatrixDecomp}
	\bm{G} = \bm{R} \bm{D} \bm{T}^{H}, 
\end{align}
where $\bm{D}$ indicates a matrix incorporating the coefficients $\gamma_{pq}$. 

Substituting \eqref{eq:GmatrixDecomp} into \eqref{eq:CommunicationModel}, and exploiting \eqref{eq:PatternMatProperty}, we reformulate the communication model as
\begin{align}
	\label{eq:CommunicationModelReformulate}
	{\bm{c}} = \frac{\eta}{2 \lambda} \bm{D} \bm{a} + s_{R} \bm{R}^{H} \bm{w},
\end{align}
Since in our H-MIMO communication system, the receive symbol vector $\bm{c}$ is an estimate of the transmit symbol vector $\bm{a}$, therefore, $P = Q$, and $\bm{D}$ should be a diagonal matrix such that the communication model \eqref{eq:CommunicationModelReformulate} forms a point-to-point mapping between each element of $\bm{c}$ and $\bm{a}$. As a consequence, $\gamma_{pq} = 0$ for $p \neq q$, leading to the newly bilinear decomposition 
\begin{align}
	\label{eq:GreenFuncBilinearDecompReformulate}
	{{\bm{G}}\left( {{{\bar{\bm{r}}}_m},{{\bar{\bm{t}}}_n}} \right)} &= \sum_{p=1}^{P} \gamma_{pp} \bm{v}_{p}(\bar{\bm{r}}_m)
	\bm{u}_{p}^{H}(\bar{\bm{t}}_n).
\end{align}

\subsection{EM-Domain Capacity Limit}

The definition of transmit and receive patterns facilitates the modeling of the end-to-end communication system, which provides a feasible access for analyzing the capacity limit of H-MIMO systems with arbitrary surface placements. In a consistent way, the established decomposition of channel matrix allows us to proceed this analysis in a tractable manner. 

Suppose that the average transmit power is $\frac{\mathcal{P}_{t}}{P}$ with $\mathcal{P}_{t}$ being the total transmit power. Then the capacity is derived based upon \eqref{eq:GmatrixDecomp} as \cite{Goldsmith2003Capacity}
\begin{align}
	\label{eq:CapacityReformulate}
	C = \log _{2} \left(\left| \bm{I}_{P} + \mu \cdot \mathrm{SNR} \cdot \bm{D} \bm{D}^{H} \right|\right),
\end{align}
where $\mu = \frac{\eta^2}{4 \lambda^2}$ and $\mathrm{SNR} = \frac{\mathcal{P}_{t}}{P s_{R} \sigma_w^2}$ denotes the average transmit signal-to-noise ratio (SNR) per unit area. Furthermore, $\mu \cdot \mathrm{SNR} \cdot \bm{D} \bm{D}^{H}$ indicates the average receive SNR per unit area. Exploiting the diagonal structure of $\bm{D}$, we can further obtain \eqref{eq:CapacityReformulate} as 
\begin{align}
	\nonumber
	C 
	&= \log_{2} \left( {\prod\nolimits_{p = 1}^P {\left( {1 + \mu \cdot \mathrm{SNR} \cdot \gamma_{pp}^2} \right)} } \right) \\
	\label{eq:CapacityReformulate1}
	&= \sum\nolimits_{p = 1}^P {\log_{2} \left( {1 + \mu \cdot \mathrm{SNR} \cdot \gamma_{pp}^2} \right)}.
\end{align}
As seen accordingly from the result, the capacity is influenced by the average receive SNR, the wavelength, and the channel related parameters, $\eta$ and $\gamma_{pp}$. 
Since we do not put our emphasis on designing the patterns, it is thus not explicit to us what the pattern matrices are. As such, $\gamma_{pp}$, associating with those patterns, cannot be directly determined, and thereby failing to unveil the capacity limit. One can refer to \cite{Sanguinetti2021Wavenumber, Zhang2022Pattern_1} for pattern design.

To mitigate the gap, we intend to transform the unknown $\gamma_{pp}$ to some known replacements, based on which, we suggest the following theorem that states the capacity limit.
\begin{theorem}
	\label{Theorem1}
	Suppose that an H-MIMO communication system transmit between surfaces with arbitrary surface placements, the near-field LoS channel model between each $mn$-pair is captured using \eqref{eq:MP2MP-Channel-4}, as well as the transmit and receive patterns that are designed by $\bm{T}$ and $\bm{R}$. The capacity limit of such a system is derived as
	\begin{align}
		\nonumber
		C &\le P \cdot \log_{2} \Bigg( 1 + \frac{\mu \cdot \mathrm{SNR}}{P} \cdot s_{R} s_{T} \cdot \sum\limits_{m=1}^{M} \sum\limits_{n=1}^{N}  \Bigg. \\
		\label{eq:Cupperbound1}
		&\;\;\;\; \qquad \qquad \qquad \qquad \Bigg. \left( \frac{\varepsilon_1}{{\bar{d}_{mn}^2}} + \frac{\varepsilon_2}{{\bar{d}_{mn}^4}} + \frac{\varepsilon_3}{{\bar{d}_{mn}^6}} \right) \Bigg),
	\end{align}
	where the coefficients are given by
	\begin{align}
		\varepsilon_1 &= \frac{1}{{16 {\pi ^2} k_0^0}} \left[ 3 - {\rm{trace}} \left( {\frac{{{\bar{\bm{d}}_{mn}}{\bar{\bm{d}}}_{mn}^T}}{{\bar{d}_{mn}^2}}} \right) \right], \tag{\ref{eq:Cupperbound1}{a}} \label{eq:coefficients(a)} \\
		\nonumber
		\varepsilon_2 &= \frac{1}{{16 {\pi ^2} k_0^2}} \left[ {5 \cdot {\rm{trace}} \left( {\frac{{{\bar{\bm{d}}_{mn}}{\bar{\bm{d}}}_{mn}^T}}{{\bar{d}_{mn}^2}}} \right) - 3} \right], \tag{\ref{eq:Cupperbound1}{b}} \label{eq:coefficients(b)} \\
		\nonumber
		\varepsilon_3 &= \frac{1}{{16 {\pi ^2} k_0^4}} \left[ {3 \cdot {\rm{trace}} \left( {\frac{{{\bar{\bm{d}}_{mn}}{\bar{\bm{d}}}_{mn}^T}}{{\bar{d}_{mn}^2}}} \right) + 3} \right]. \tag{\ref{eq:Cupperbound1}{c}} \label{eq:coefficients(c)}
	\end{align}
\end{theorem}
\begin{proof}
	See appendix \ref{Theorem1Proof}.
\end{proof}

\textit{Remark:} It can be seen that the capacity limit grows logarithmically with the product of transmit element area $s_{T}$, receive element area $s_{R}$, and the combined effects of $1/{\bar{d}_{mn}^2}$, $1/{\bar{d}_{mn}^4}$, and $1/{\bar{d}_{mn}^6}$ over all $M$ and $N$ antenna elements. Alternatively, it decreases logarithmically with the combined effects of ${\bar{d}_{mn}^2}$, ${\bar{d}_{mn}^4}$, and ${\bar{d}_{mn}^6}$ over all $M$ and $N$ antenna elements. These powers of $\bar{d}_{mn}$ serve as a guideline for system designs. Specifically, ${\bar{d}_{mn}^6}$ dominates the capacity limit in the near-field region, whereas ${\bar{d}_{mn}^2}$ primarily affects the capacity limit in far-field regions. This offers coarse but pragmatic estimates of the performance limits for practical systems.


The result can be generalized to far-field scenarios as a special case, which can reversely validate the effectiveness of our result. Those terms corresponding to ${\bar{d}_{mn}^4}$, and ${\bar{d}_{mn}^6}$ tend to be attenuated in far-field regions. 
Furthermore, the center distance between each $mn$-pair, $\bar{d}_{mn}$, can be approximated to a common distance, denoted as $d_{0}$, irrelevant to index $m$ and $n$, due to $d_{0} \gg \max\{A_{T}, A_{R}\}$. The direction (or unit vector) of the $mn$-pair, $\bar{\bm{d}}_{mn}/\bar{d}_{mn}$, tends to be unified to a certain common direction, represented by $\bm{\kappa}$. 
As a consequence, the capacity limit becomes
\begin{align}
	\label{eq:CupperboundFarField}
	C \le P \cdot \log_{2} \left( 1 + \frac{\mu \cdot \mathrm{SNR}}{P} \cdot \bar{\varepsilon}_1 \cdot \frac{A_{R} A_{T}}{{d_{0}^{2}}} \right),
\end{align}
where $\bar{\varepsilon}_1$ is deduced from $\varepsilon_1$ based on $\left( \bar{\bm{d}}_{mn} \bar{\bm{d}}_{mn}^{T} \right) / {{\bar{d}_{mn}^2}} \approx \bm{\kappa} \bm{\kappa}^{T}$ and $\bm{\kappa}^{T} \bm{\kappa} = 1$, given by
\begin{align}
	\label{eq:varepsilon}
	\bar{\varepsilon}_1 = \frac{1}{{16 {\pi ^2} }} \left[ 3 - {\rm{trace}} \left( \bm{\kappa} \bm{\kappa}^{T} \right) \right] = \frac{1}{{8 {\pi ^2} }},
\end{align}
$A_{R} = M s_{R}$ and $A_{T} = N s_{T}$ are the overall surface area of RX and TX, respectively. It can be seen that the capacity limit in the far-field scenario grows logarithmically with the transmit and receive surface areas, $A_{T}$ and $A_{R}$, and it decreases logarithmically with the square of TX-RX distance. Our conclusion of the far-field case is consistent with the insights on DoF and power coupling in \cite{Dardari2020Communicating2, Miller2000Communicating}. 

\section{Numerical Evaluations}
\label{SectionNR}

We present numerical evaluations in this section, in which we first demonstrate the accuracy of our established channel models in capturing the essence of the wireless channel, and we then exhibit the capacity limit of the H-MIMO system using our derived results. 
Unless specifically stated, we focus on square surfaces, corresponding to the elevation angles and the azimuth angles of TX and RX as $\theta_h = \theta_v = 90^{o}$, $\phi_{h} = 0^{o}$, $\phi_{v} = 90^{o}$, and $\vartheta_h = 90^{o}$, $\vartheta_v = 90^{o}$, $\psi_{h} = 0^{o}$, $\psi_{v} = 90^{o}$, respectively. 
Furthermore, evaluations are tested at $2.4$ GHz, corresponding to $\lambda = 0.125$ meters.
In addition, to further discriminate the near-field region and far-field region, we employ the Rayleigh distance, defined as $d_{Rayleigh} = \frac{2 (D_{TX} + D_{RX})^2}{\lambda}$, where $D_{TX}$ and $D_{RX}$ reflect the diameter of a circle (diagonal length of the square surface) that minimally covers the TX surface and the RX surface, respectively. As such, $d < d_{Rayleigh}$ indicates the near-field region, \footnote{Note that the near-field region can be further divided into the reactive near-field region and the radiative near-field region with the distance boundary given by $d_{boundary} = 0.62 \sqrt{ \frac{ (D_{TX} + D_{RX})^3}{\lambda} }$. Since the distance range of the reactive near-field region is quite small, thus in our simulations, we mainly test the TX-RX distance within the radiative near-field region.} whereas $d > d_{Rayleigh}$ reflects the far-field region. Given the element spacing (function of $\lambda$, denoted by $\Delta \lambda$ with $\Delta \le 0.2$ for metasurface antennas), we obtain $D_{TX} = \sqrt{2N}\Delta \lambda$, $D_{RX} = \sqrt{2M}\Delta \lambda$, and $d_{Rayleigh} = 4 (\sqrt{N} + \sqrt{M})^2 \Delta^2 \lambda$.

\subsection{Channel Model Evaluation}

In channel model evaluations, two metrics are mainly focused, namely, the normalized mean-squared error (NMSE) between our established channel models and the true wireless channel, and the eigenvalues of the channel matrix in capturing eigenmodes of the true wireless channel. The eigenvalues can exactly reveal the eigenmodes (number of available eigenvalues larger than some thresholds) of H-MIMO systems.
Particularly, the NMSE is defined as: 
\begin{align}
	\label{eq:normalizedMSE}
	{\rm{NMSE}} = \frac{\| \hat{\bm{H}} - \bm{H} \|_{F}^{2}}{\| \bm{H} \|_{F}^{2}},
\end{align}
where $\bm{H}$ and $\hat{\bm{H}} \in \{ \bm{H}^{\text{CD-CM}}, \bm{H}^{\text{CI-CM}} \}$ are the channel matrix obtained by INT-CM and the modeled channel matrix, respectively. 
Moreover, we specify: 
\begin{itemize}
	\item 
	``INT-CM" indicates $\bm{H}$ obtained via \eqref{eq:MP2MP-Channel};
	\item 
	``CD-CM" specifies $\bm{H}^{\text{CD-CM}}$ obtained via \eqref{eq:MP2MP-Channel-3};
	\item 
	``CI-CM" reveals $\bm{H}^{\text{CI-CM}}$ obtained via \eqref{eq:MP2MP-Channel-4}.
\end{itemize}
We provide numerical results with respect to element spacing, TX-RX distance, and the number of TX elements. 
Evaluations are carried out over three angle configurations $\vartheta_v = \{60^{o}, 75^{o}, 90^{o}\}$.

\begin{figure}[t!]
	\centering
	\includegraphics[height=6.4cm, width=7.8cm]{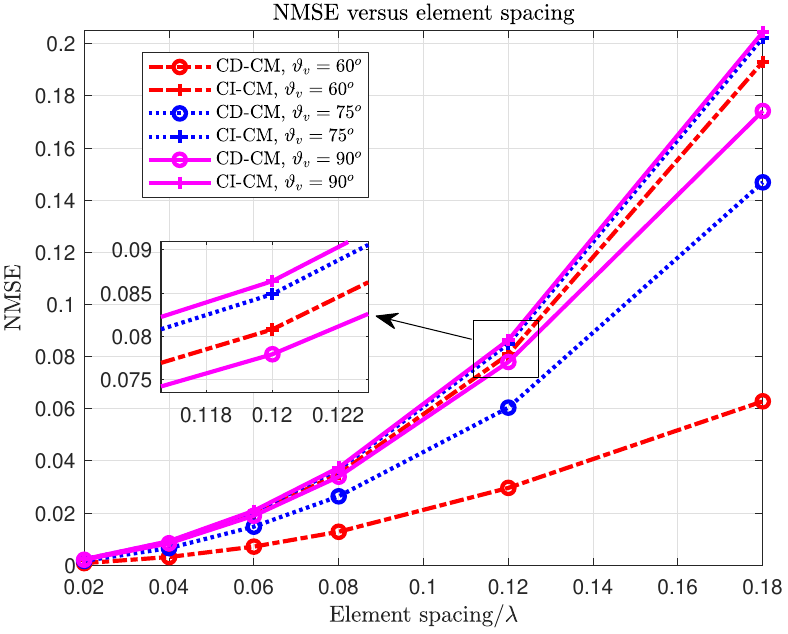}
	\vspace{-0.5em}
	\caption{NMSE of established channel models versus the element spacing of antenna surfaces.}
	\label{fig:Main_NMSE_Spacing_3Angles_Alter_Comparison}
\end{figure}

\begin{figure}[t!]
	\centering
	\includegraphics[height=6.4cm, width=7.8cm]{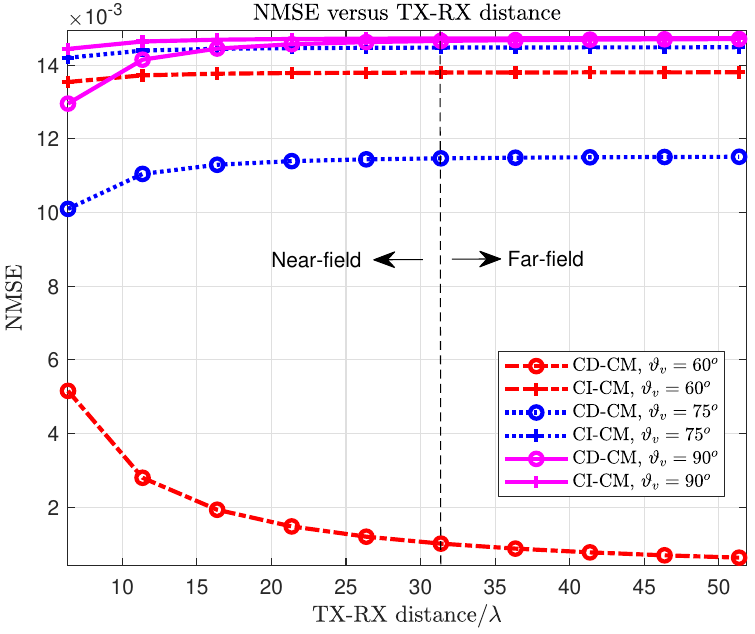}
	\vspace{-0.5em}
	\caption{NMSE of established channel models versus the TX-RX distance.}
	\label{fig:Main_NMSE_Distance_3Angles_Comparison}
\end{figure}

We first present the NMSE variation of our channel models with respect to the element spacing. We keep the Rayleigh distances for all spacing cases the same, and set the TX-RX distance less than the Rayleigh distance. Specifically, we set $d_{Rayleigh} = 8.2944 \lambda = 1.0368$ meters by selecting $\sqrt{N} = 51$, $\sqrt{M} = 21$ and $\Delta = 0.02$ (corresponding to $D_{TX} \approx 0.18$ meters and $D_{RX} \approx 0.074$ meters). \footnote{We select these parameters based on the computing capability of our computer. Note that one can expand the Rayleigh distance to show a more practical communication distance, e.g., $5$ meters, by enlarging $\sqrt{N}$ and $\sqrt{M}$ if a computer with strong computing capability is available.} $\Delta = 0.02, 0.04, 0.06, 0.08, 0.12, 0.18$ are compared for different element spacing configurations, and thereby resulting in $(\sqrt{N}, \sqrt{M}) = (51, 21)$, $(25, 11)$, $(17, 7)$, $(13, 5)$, $(9, 3)$, $(7, 1)$. \footnote{Note that the values of $\Delta$ are selected such that $(\sqrt{N}, \sqrt{M})$ are guaranteed to be integers.} 
The comparison results are shown in Fig. \ref{fig:Main_NMSE_Spacing_3Angles_Alter_Comparison}, where we have three observations: \ding{192} The NMSE of CD-CM are always smaller than those of CI-CM, showing its better approximation to the INT-CM. \ding{193} CD-CM and CI-CM become closely in NMSE as $\vartheta_v$ varying from $60^{o}$ to $90^{o}$, namely as the TX-RX surfaces tending to be parallel. \ding{194} As the element spacing becoming small, the NMSE of CD-CM and CI-CM drastically decrease, showing a good agreement in capturing the nearly continuous surface induced H-MIMO channel.

We then show the NMSE variation of the proposed channel models in terms of the TX-RX distance. We specify $\sqrt{N} = 41$, $\sqrt{M} = 15$, and $\Delta = 0.05$ for an element spacing of $0.05 \lambda$, which results in $D_{TX} \approx 0.36$ meters, $D_{RX} \approx 0.13$ meters, and $d_{Rayleigh} = 31.36 \lambda = 3.92$ meters (a practical near-field communication distance for indoor applications). We vary the TX-RX distance from $6.36 \lambda = 0.795$ meters to $51.36 \lambda = 6.42$ meters, and show the simulation results in Fig. \ref{fig:Main_NMSE_Distance_3Angles_Comparison}. 
It can be seen that as the distance becoming large, the NMSE of CD-CM and CI-CM tend to be flat, which are mainly restricted by the element spacing. 
In addition, CD-CM is superior to CI-CM in the NMSE sense, and they get closely in the NMSE sense for different $\vartheta_v$ varying from $60^{o}$ to $90^{o}$.

\begin{figure}[t!]
	\centering
	\includegraphics[height=6.4cm, width=7.8cm]{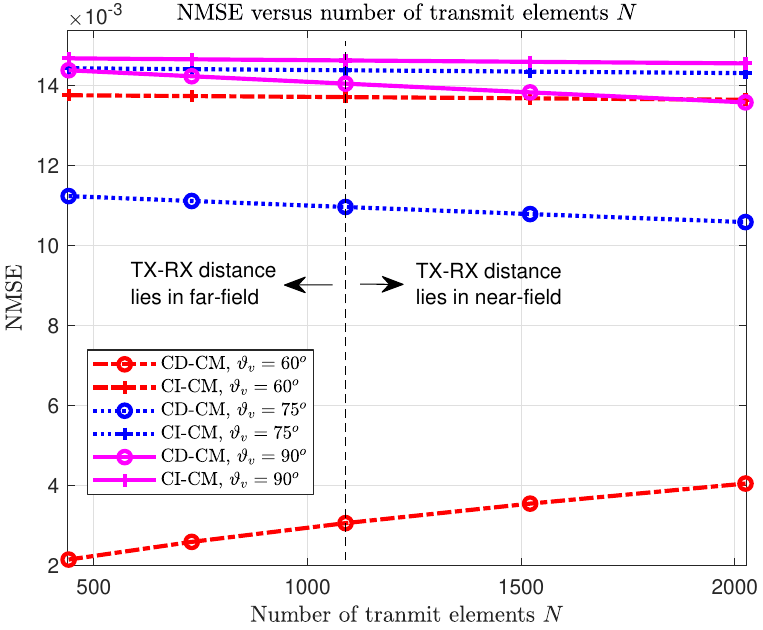}
	\vspace{-0.5em}
	\caption{NMSE of established channel models versus the number of transmit elements $N$.}
	\label{fig:Main_NMSE_TXElementNum_3Angles_Comparison}
\end{figure}

We further demonstrate the NMSE with respect to the number of transmit elements $N$ as demonstrated in Fig. \ref{fig:Main_NMSE_TXElementNum_3Angles_Comparison}. We fix $\Delta = 0.05$ to get the element spacing as $0.05 \lambda$, and vary the number of transmit elements in the range of $N \in [21 \times 21, 45 \times 45]$ (corresponding to $D_{TX} \in [0.1856, 0.3977]$ meters) while fixing the number of receive elements to $M = 15 \times 15$ (corresponding to $D_{RX} = 0.1326$ meters). We fix the TX-RX distance as the Rayleigh distance $d_{Rayleigh} = 23.04 \lambda = 2.88$ meters obtained under the value of $N = 33 \times 33$, and the above values of $\Delta$ and $M$. 
We see from Fig. \ref{fig:Main_NMSE_TXElementNum_3Angles_Comparison} that: \ding{192} The NMSE of CD-CM is smaller than that of CI-CM for all tested number of transmit elements, showing a better depiction of CD-CM once again. \ding{193} When the TX and RX surfaces tend to be parallel ($\vartheta_v = 75^{o}, 90^{o}$),  the NMSE of CD-CM decreases as the number of transmit elements increases, while for a more non-parallel case ($\vartheta_v = 60^{o}$), the NMSE of CD-CM increases slightly. \ding{194} For all tested setups, the NMSE of CI-CM becomes flat.

Next, we evaluate the proposed channel models in depicting the eigenvalues of the channel matrix obtained via the INT-CM. The element spacing is selected as $0.05 \lambda$, and $\vartheta_v$ is fixed as $90^{o}$ with other angles unchanged. 
We set the number of transmit/receive elements as $N = 35 \times 35, M = 9 \times 9$ and $N = 41 \times 41, M = 15 \times 15$, respectively. 
We demonstrate the eigenvalues for three different TX-RX distances, namely, $2.4427 \lambda$, $10.1803 \lambda$, and $39.2 \lambda$ (corresponding to $0.3$ meters, $1.27$ meters, and $4.9$ meters, respectively). It is noted that the first two distances are less than the Rayleigh distance obtained with $N = 41 \times 41, M = 15 \times 15$, and the last distance exceeds this Rayleigh distance. 

First, we observe from Fig. \ref{fig:Main_EigenValue_Demon}(a) that when the TX-RX distance is $2.4427 \lambda$ ($0.3$ meters), even if the extremely short distance, CD-CM and CI-CM can well fit the INT-CM in its eigenvalues and eigenmodes. 
Afterwards, comparing among Fig. \ref{fig:Main_EigenValue_Demon}(a)(b)(c), as the TX-RX distance varies from $2.4427 \lambda$ to $10.1803 \lambda$ and $39.2 \lambda$, we see that both CD-CM and CI-CM fit the INT-CM more and more accurate in eigenvalues and eigenmodes.  Moreover, for shorter distances, more eigenmodes are introduced than those of longer distance, contributing to more independent transmission channels, even in LoS cases. 
From each individual demonstration of Fig. \ref{fig:Main_EigenValue_Demon}(a)(b)(c), we get that a large number of transmit/receive elements facilitate more eigenmodes than those of a small number of transmit/receive elements within a proper range of distance.

\begin{figure}[!tbp]
	\centering
	\subfloat[]{\label{fig:Main_EigenValue_Demon_1}\includegraphics[width=0.9\columnwidth]{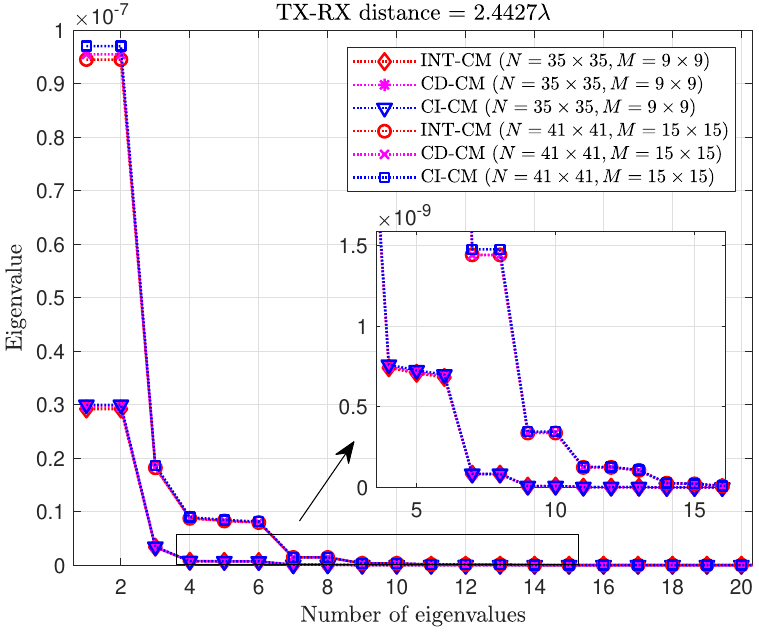}} \\
	\vspace{1em}
	\subfloat[]{\label{fig:Main_EigenValue_Demon_2}\includegraphics[width=0.9\columnwidth]{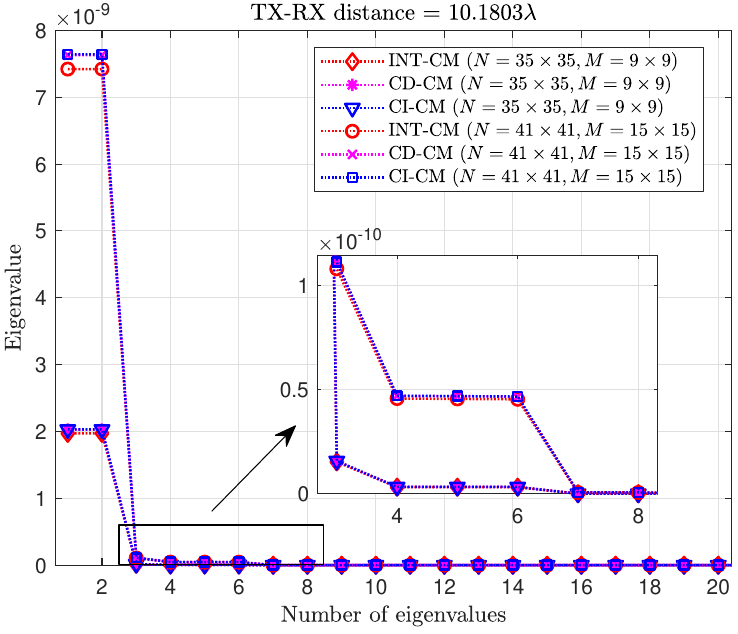}} \\
	\vspace{1em}
	\subfloat[]{\label{fig:Main_EigenValue_Demon_3}\includegraphics[width=0.9\columnwidth]{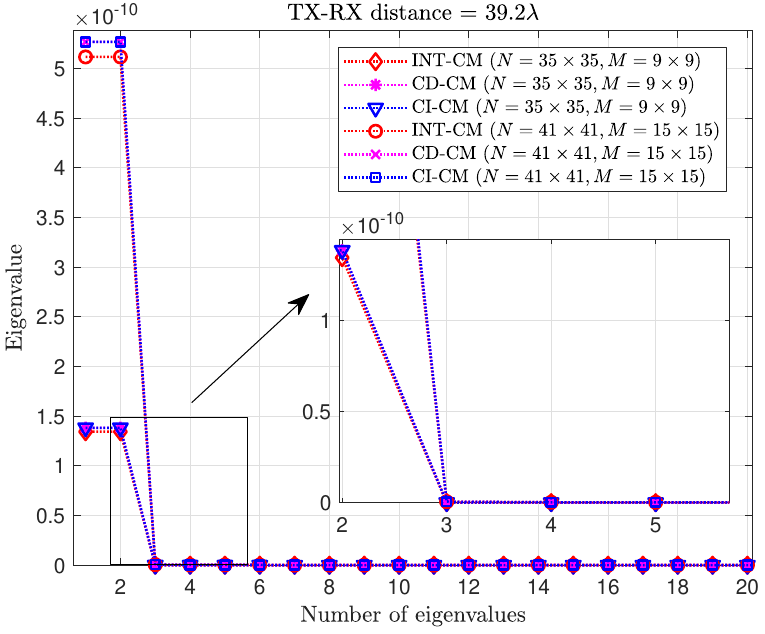}} \\
	\caption{Eigenvalues of channel matrices with respect to different transmit and receive elements at the TX-RX distance of (a) $2.4427 \lambda$, (b) $10.1803 \lambda$, and (c) $39.2 \lambda$, respectively.}
	\label{fig:Main_EigenValue_Demon}
\end{figure}

The above evaluations show the effectiveness of our channel models in capturing the essence of the true wireless channel.

\subsection{Capacity Limit Evaluation}

In capacity limit evaluations, we focus on demonstrating the effectiveness and tightness of our derived upper bound in depicting the exact capacity limit. The numerical evaluations are performed in terms of the average transmit SNR, the element spacing, the TX-RX distance, and the number of transmit elements. 
Numerical assessments of the exact capacity \eqref{eq:CapacityReformulate1} incorporate the evaluation of the unknown $\gamma_{pp}$, which needs to be addressed. For this purpose, we resort to the singular value decomposition via $\bm{G} = \bm{U} \bm{\Sigma} \bm{V}^{H}$, and set $\bm{R} = \frac{\bm{U}(:, 1:P)}{\sqrt{s_{R}}}$, $\bm{T} = \frac{\bm{V}(:, 1:P)}{\sqrt{s_{T}}}$. As a consequence, $\bm{D} = \sqrt{s_{R} s_{T}} \bm{\Sigma}(1:P, 1:P)$, and $\gamma_{pp}$ can be determined accordingly using  singular values of $\bm{\Sigma}$. Moreover, we only select effective singular values from all noise corrupted ones, i.e., top-$P$ largest values, where $P$ is determined adaptively when the top-$P$ values account for no lower than $95\%$ power of all the singular values. 
We define ``Upper bound" to indicate \eqref{eq:Cupperbound1}, and use ``Upper bound (FF)" to specify \eqref{eq:CupperboundFarField} for the far-field case.

\begin{figure}[t!]
	\centering
	\includegraphics[height=6.4cm, width=7.8cm]{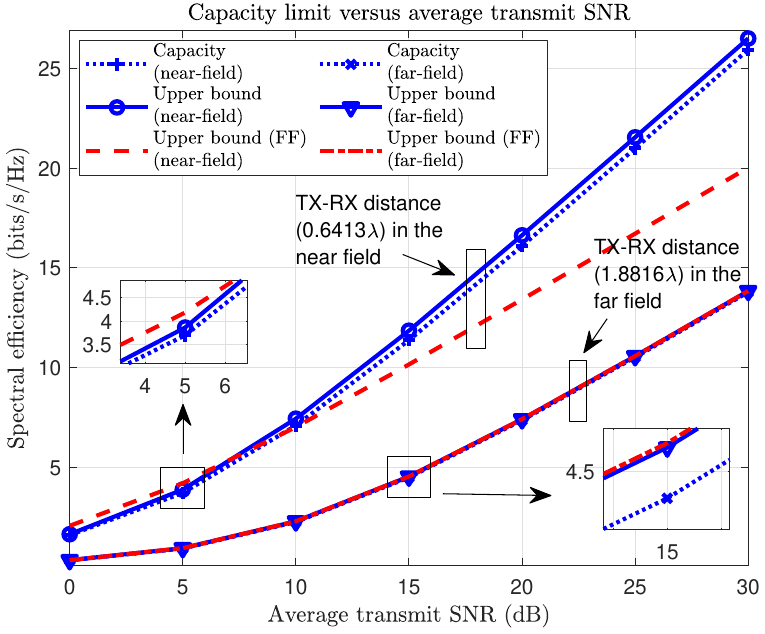}
	\vspace{-0.5em}
	\caption{Capacity limit versus the TX average SNR.}
	\label{fig:AvePowerAlloc_Main_Capacity_SNR_3DistancesComp}
\end{figure}

We check the capacity limit versus the average transmit SNR in Fig. \ref{fig:AvePowerAlloc_Main_Capacity_SNR_3DistancesComp}, in which we present numerical results for different TX-RX distances, $0.6413 \lambda$ and $1.8816 \lambda$ (corresponding to $0.08$ meters and $0.2352$ meters \footnote{These distances can be extended to larger values for simulating larger $N$ and $M$, which we omit due to high computational complexity. The use of short distances does not influence the correctness of our result.}), under the configurations of $N = 41 \times 41$, $M = 15 \times 15$, and $\Delta = 0.01$ (corresponding to $D_{TX} = 0.0725$ meters and $D_{RX} = 0.0265$ meters). It is noted that the TX-RX distances belong to the near-field region and the far-field region, respectively. 
From the figure, we first notice that the capacity in the near-field region is higher than that in the far-field region, and it is proportional to the SNR. We then observe that the ``Upper bound" provides an effective and tightness bound for the exact capacity in both near-field and far-field cases. Another notable phenomenon is that the ``Upper bound (FF)" fails to offer a valid bound on the exact capacity, especially when the distance is small and falls in the near-field region.

\begin{figure}[t!]
	\centering
	\includegraphics[height=6.4cm, width=7.8cm]{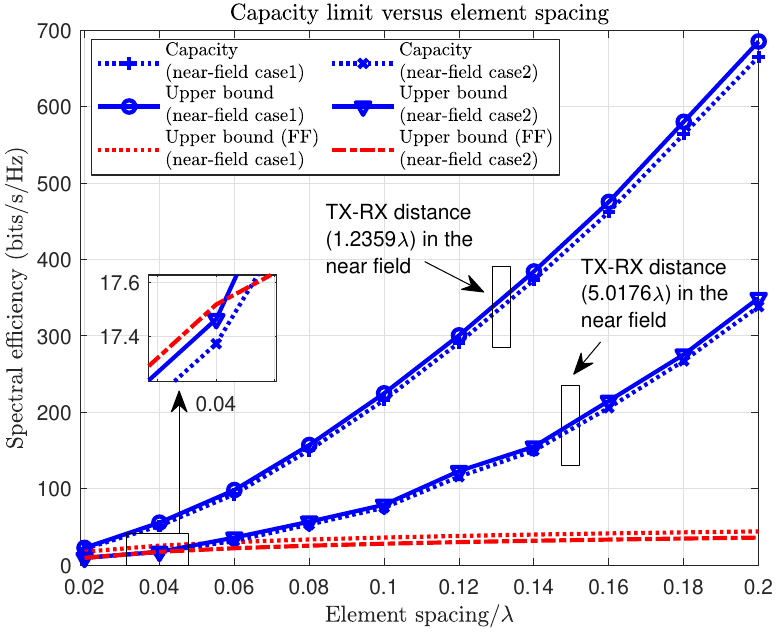}
	\vspace{-0.5em}
	\caption{Capacity limit versus the element spacing.}
	\label{fig:AvePowerAlloc_Main_Capacity_Spacing_3DistancesComp}
\end{figure}

Subsequently, the capacity limit versus the element spacing is shown in Fig. \ref{fig:AvePowerAlloc_Main_Capacity_Spacing_3DistancesComp}, where we evaluate the element spacing varying from $0.02 \lambda$ to $0.2 \lambda$, and apply the same setup in $N$ and $M$, as those in Fig. \ref{fig:AvePowerAlloc_Main_Capacity_SNR_3DistancesComp}. For the TX-RX distances, we evaluate $1.2359 \lambda$ and $5.0176 \lambda$ ($0.1545$ meters and $0.6272$ meters), which fall within the near-field region for all tested element spacing. 
Therefore, as the number of transmit and receive elements are fixed, the element spacing determines the surface area, i.e. a large/small element spacing results in a large/small surface area. From Fig. \ref{fig:AvePowerAlloc_Main_Capacity_Spacing_3DistancesComp}, we observe that the capacity increases with the element spacing. In addition, for all tested cases, ``Upper bound" provides a tightness upper edge on the exact capacity limit in the near-field region, while ``Upper bound (FF)" fails to capture the accurate upper bound and the exact capacity limit.

\begin{figure}[t!]
	\centering
	\includegraphics[height=6.4cm, width=7.8cm]{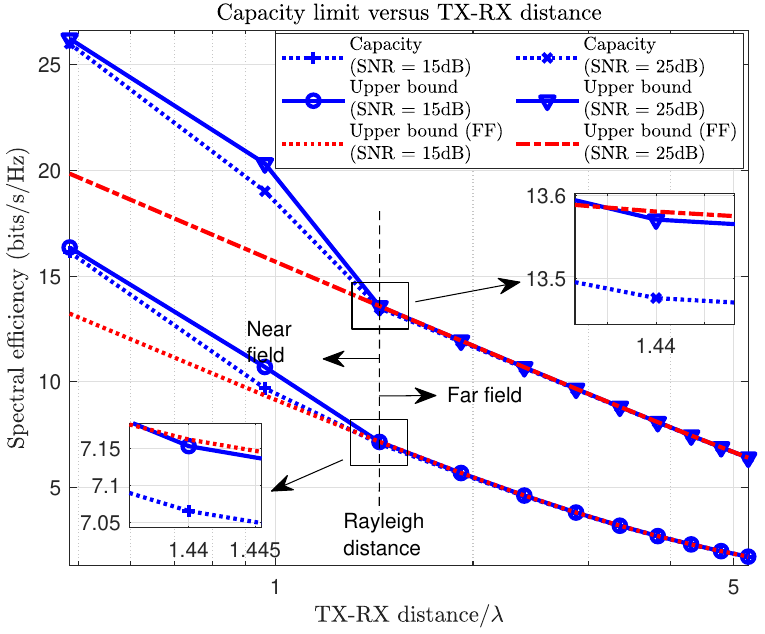}
	\vspace{-0.5em}
	\caption{Capacity limit versus the TX-RX distance.}
	\label{fig:AvePowerAlloc_Main_Capacity_Distance_3SNRsComp}
\end{figure}

\begin{figure}[t!]
	\centering
	\includegraphics[height=6.4cm, width=7.8cm]{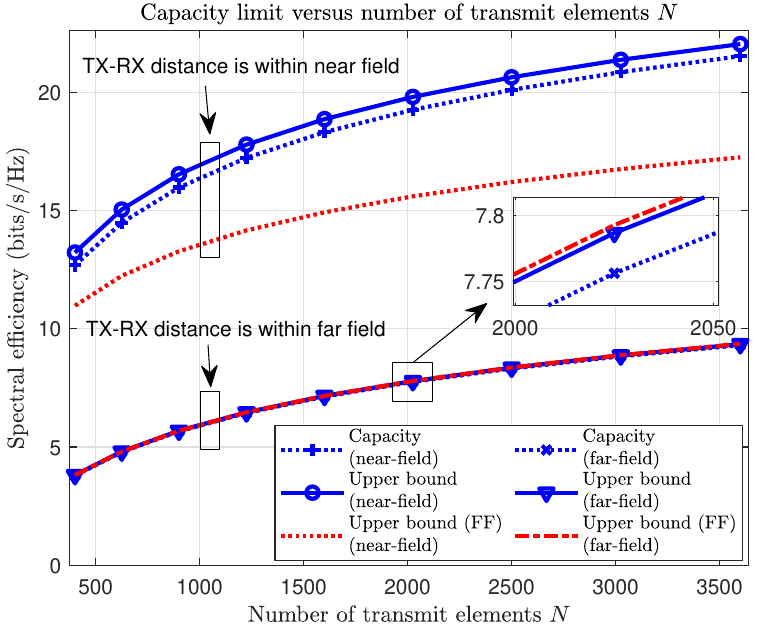}
	\vspace{-0.5em}
	\caption{Capacity limit versus number of transmit elements $N$.}
	\label{fig:AvePowerAlloc_Main_Capacity_TXElementNum_3DistancesComp}
\end{figure}


Furthermore, we unveil the capacity limit versus the TX-RX distance. The parameters are configured as $N = 40 \times 40$, $M = 20 \times 20$, and the element spacing $0.01 \lambda$, resulting in $D_{TX} = 0.0707$ meters, $D_{RX} = 0.0354$ meters, and the Rayleigh distance $d_{Rayleigh} = 1.44 \lambda = 0.18$ meters. 
In Fig. \ref{fig:AvePowerAlloc_Main_Capacity_Distance_3SNRsComp}, we demonstrate the capacity for two SNR values (15 dB and 25 dB) and TX-RX distances from $0.4846 \lambda$ to $5.2616 \lambda$, encompassing the near-field and the far-field regions.
It can be seen that, as the TX-RX distance becomes large, the capacity limits drop drastically. The ``Upper bound" successfully depicts the exact capacity limit within a small gap over both the near-field region and the far-field region, and it gradually collapses to the ``Upper bound (FF)" as the TX-RX distance increases and exceeds the Rayleigh distance.

In the end, the influence of the number of transmit elements $N$ on capacity limit is exhibited, where we fix $M = 20 \times 20$ and vary $N$ from $20 \times 20$ to $60 \times 60$ with the element spacing $0.01 \lambda$. Two TX-RX distances are selected for demonstration, which are given by the Rayleigh distance obtained when $N = 20 \times 20$ and $N = 60 \times 60$, respectively. 
In the former case, the system works in the near-field region and in the latter case, in the far-field region, for all values $N$, as annotated in Fig. \ref{fig:AvePowerAlloc_Main_Capacity_TXElementNum_3DistancesComp}. Fig. \ref{fig:AvePowerAlloc_Main_Capacity_TXElementNum_3DistancesComp} shows that in the near-field case, the exact capacity limit can be perfectly captured by the ``Upper bound" when varying $N$. However, the ``Upper bound (FF)" fails to capture the capacity. Moreover, in the far-field case, it is obvious that the ``Upper bound" gets close to the ``Upper bound (FF)", and both are capable of depicting the exact capacity, which proves the ``Upper bound (FF)" is a far-field special case of the ``Upper bound". 
In addition, one can also observe that the capacity increases as $N$ gets larger, and the near-field case offers a higher capacity than the far-field case.

Overall, the evaluations reveal that our derived capacity upper bound is effective and quite tight for depicting the real capacity limit. Besides, the large surface area and the short distance facilitate the increase in H-MIMO capacity.

\section{Conclusions}
\label{SectionCON}
In this article, we considered the point-to-point H-MIMO systems with arbitrary surface placements in a near-field LoS scenario, in which we established the generalized EM-domain near-field LoS channel models and studied the capacity limit. We first established effective, explicit, and computationally-efficient CD-CM and CI-CM, which are valid in approaching the integral form near-field LoS channel and in capturing the essence of physical wireless channel, such as the DoF of channel matrix. We then built an effective analytical framework for deriving the capacity limit. We showed that the capacity limit grows logarithmically with the product of TX and RX element areas and the combined effects of $1/{\bar{d}_{mn}^2}$, $1/{\bar{d}_{mn}^4}$, and $1/{\bar{d}_{mn}^6}$ over all $M$ and $N$ antenna elements. Our result can exactly capture the exact capacity, offering an effective mean for predicting the system performance.

\appendices

\section{Proof of Lemma \ref{LemmaSurfaceShape}}
\label{LemmaSurfaceShapeProof}

Without loss of generality, we take the TX surface for example. Suppose that the TX surface has a horizontal length $L_T^h$, a vertical length $L_T^v$, and a diagonal length $L_T$, as shown in Fig. \ref{fig:LemmaSurfaceShape}. The projections of horizontal and vertical lengths on the $xy$-plane are defined as $\ell_T^h$ and $\ell_T^v$, respectively, which can be derived as $\ell_T^h = L_T^h \sin \theta_{h}$ and $\ell_T^v = L_T^v \sin \theta_{v}$.
Likewise, we can also obtain the projections of horizontal and vertical lengths on the $z$-axis as $z_T^h = L_T^h \cos \theta_{h}$ and $z_T^v = L_T^v \cos \theta_{v}$. 
Therefore, we know that $\left( L_T^h \right)^{2} = \left( \ell_T^h \right)^{2} + \left( z_T^h \right)^{2}$, and $\left( L_T^v \right)^{2} = \left( \ell_T^v \right)^{2} + \left( z_T^v \right)^{2}$. Note that if the surface (and element) is rectangle (or square), the following equation should hold using the Pythagorean theorem
\begin{align}
	\nonumber
	\left( L_T \right)^{2} 
	&= \left( L_T^h \right)^{2} + \left( L_T^v \right)^{2} \\
	\nonumber
	&= \left( \ell_T^h \right)^{2} + \left( L_T^h \cos \theta_{h} \right)^{2} 
	+ \left( \ell_T^v \right)^{2} + \left( L_T^v \cos \theta_{v} \right)^{2}.
\end{align}
Alternatively, we can express $L_T$ as
\begin{align}
	\nonumber
	\left( L_T \right)^{2} 
	&= \left( \ell_T^h \right)^{2}
	+ \left( \ell_T^v \right)^{2} + \left( z_T \right)^{2} \\
	\nonumber
	&\overset{(1)}{=} \left( \ell_T^h \right)^{2}
	+ \left( \ell_T^v \right)^{2} + \left( L_T^h \cos \theta_{h} + L_T^v \cos \theta_{v} \right)^{2},
\end{align}
where $\overset{(1)}{=}$ holds by noticing that $z_T = z_T^h + z_T^v$. Equaling both expressions of $L_T^{2}$, we get the following equation that should be satisfied if the surface (and element) is rectangle (or square)
\begin{align}
	\nonumber
	\left( L_T^h \cos \theta_{h} \right)^{2} 
	&+ \left( L_T^v \cos \theta_{v} \right)^{2} 
	= \left( L_T^h \cos \theta_{h} + L_T^v \cos \theta_{v} \right)^{2} \\
	\nonumber
	&\Rightarrow \left( L_T^h \cos \theta_{h} \right) \left( L_T^v \cos \theta_{v} \right) = 0.
\end{align}
It is obvious that $\theta_{h}$ and $\theta_{v}$ should be the angle that guarantees either $\cos \theta_{h}$ or $\cos \theta_{v}$ to be zero. Considering the practical placement of the surface, both angles are within the range of $[0, 180^o]$, such that we obtain ${{\theta _h} = {90^o} \; {\rm{or}} \; {\theta _v} = {90^o}}$. We can similarly derive the results for RX as ${{\vartheta _h} = {90^o} \; {\rm{or}}  \; {\vartheta _v} = {90^o}}$, which completes the proof.

\begin{figure}[t!]
	\centering
	\includegraphics[height=5.4cm, width=6.8cm]{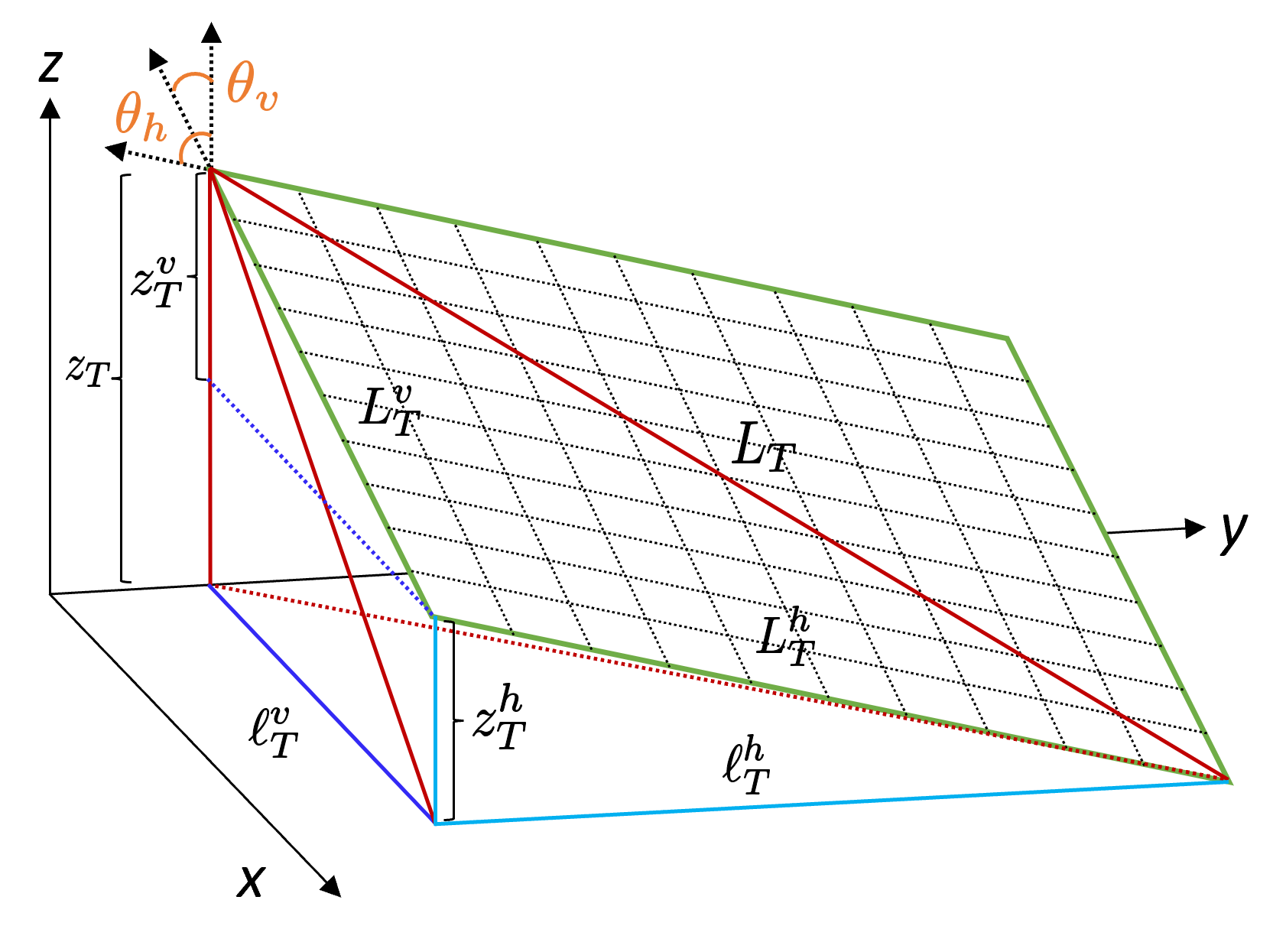}
	\caption{TX Surface in Cartesian coordinates.}
	\label{fig:LemmaSurfaceShape}
	\vspace{-0.5em}
\end{figure}

\section{Proof of Lemma \ref{LemmaZbyXY}}
\label{LemmaZbyXYProof}
For any arbitrary $\Delta {\bm{t}}_n = [\Delta x_n \; \Delta y_n \; \Delta z_n]^{T}$, it can be represented by the combination of the horizontal unit vector ${\bm{t}}_T^h = [\sin \theta_{h} \cos \phi_{h} \; \sin \theta_{h} \sin \phi_{h} \; \cos \theta_{h}]^{T}$ and the vertical unit vector ${\bm{t}}_T^v = [\sin \theta_{v} \cos \phi_{v} \; \sin \theta_{v} \sin \phi_{v} \; \cos \theta_{v}]^{T}$ with weights, i.e., $\alpha$ and $\beta$, respectively. We thus have
\begin{align}
	\nonumber
	\Delta {\bm{t}}_n = \alpha {\bm{t}}_T^h + \beta {\bm{t}}_T^v 
	= \left[\begin{array}{l}
		\alpha \sin \theta_{h} \cos \phi_{h}+\beta \sin \theta_{v} \cos \phi_{v} \\
		\alpha \sin \theta_{h} \sin \phi_{h}+\beta \sin \theta_{v} \sin \phi_{v} \\
		\alpha \cos \theta_{h}+\beta \cos \theta_{v}
	\end{array}\right].
\end{align}
In order to derive $\Delta z_n$, represented by $\Delta x_n$ and $\Delta y_n$, we need to first derive $\alpha$ and $\beta$. 
To this aim, we multiply $\Delta x_n$ with $\sin \phi_{v}$, and $\Delta y_n$ with $\cos \phi_{v}$, respectively, yielding 
\begin{align}
	\nonumber
	\Delta x_n \sin \phi_{v} &= \alpha \sin \theta_{h} \cos \phi_{h} \sin \phi_{v} + \beta {\sin } \theta_{v} \cos \phi_{v} \sin \phi_{v}, \\
	\nonumber
	\Delta y_n \cos \phi_{v} &= \alpha \sin \theta_{h} \sin \phi_{h} \cos \phi_{v} + \beta \sin \theta_{v} \sin \phi_{v} \cos \phi_{v}.
\end{align}
Conducting a subtraction, we get
\begin{align}
	\nonumber
	\Delta x_n \sin \phi_{v} &- \Delta y_n \cos \phi_{v} 
	= \alpha \sin \theta_{h} \sin \left(\phi_{v}-\phi_{h}\right), \\
	\nonumber
	&\Rightarrow \;
	\alpha = \frac{\Delta x_n \sin \phi_{v}-\Delta y_n \cos \phi_{v}}{\sin \theta_{h} \sin \left(\phi_{v}-\phi_{h}\right)}
\end{align}
when $\sin \theta_{h} \neq 0$ and $\sin \left(\phi_{v}-\phi_{h}\right) \neq 0$.
Similarly, multiplying $\Delta x_n$ with $\sin \phi_{h}$, and $\Delta y_n$ with $\cos \phi_{h}$, respectively, and then making a subtraction, we derive
\begin{align}
	\nonumber
	\Delta x_n \sin \phi_{h} &- \Delta y_n {\cos \phi_{h}} 
	= \beta \sin \theta_{v} \sin \left(\phi_{h}-\phi_{v}\right), \\
	\nonumber
	&\Rightarrow \;
	\beta = \frac{\Delta x_n \sin \phi_{h}-\Delta y_n \cos \phi_{h}}{\sin \theta_{v} \sin \left(\phi_{h}-\phi_{v}\right)}
\end{align}
for $\sin \theta_{v} \neq 0$ and $\sin \left(\phi_{h}-\phi_{v}\right) \neq 0$. 
With the availability of $\alpha$ and $\beta$, we can obtain $\Delta z_n = \alpha \cos \theta_h+\beta \cos \theta_v$ as
\begin{align}
	\nonumber
	\Delta z_n 
	&=\frac{\cot \theta_v \sin \phi_n-\cot \theta_h \sin \phi_v}{\sin \left(\phi_h-\phi_v\right)} \Delta x_n \\ 
	\nonumber
	&+ \frac{\cot \theta_v \cos \phi_h-\cot \theta_h \cos \phi_v}{\sin \left(\phi_h-\phi_v\right)} \Delta y_n.
\end{align}
Following a similar derivation process, for any arbitrary point $\Delta {\bm{r}}_m = [\Delta x_m \; \Delta y_m \; \Delta z_m]^{T} \in s_R$, $\Delta z_m$ can be represented via $\{\Delta x_m, \Delta y_m\}$, which completes the proof.

\section{Proof of Theorem \ref{TheoremChannelModel}}
\label{TheoremChannelModelProof}

Based upon Lemma \ref{LemmaZbyXY}, we can proceed to derive $I_{T}$ and $I_{R}$, respectively, where $I_{T}$ is first derived as follows. 
By expanding the integral over the surface area by integrals over the horizontal length and the vertical length, we first obtain
\begin{align}
	\nonumber
	I_{T} 
	&= \int_{s_T} e^{ - i {k_0}  \frac{\bar{\bm{d}}_{mn}^{T}}{\bar{d}_{mn}} \Delta {{\bm{t}}_n} } {\rm{d}}{\Delta {\bm{t}}_n} \\
	\nonumber
	&= \int_{-\frac{l_{T}^{v}}{2}}^{\frac{l_{T}^{v}}{2}} \int_{-\frac{l_{T}^{h}}{2}}^{\frac{l_{T}^{h}}{2}} 
	e^{ - i {k_0} \frac{ {{\bar x}_{mn}} \Delta {x_n} + {{\bar y}_{mn}} \Delta {y_n} + {{\bar z}_{mn}} \Delta {z_n} }{\bar{d}_{mn}} } 
	{\rm{d}}{\Delta {x}_n}  {\rm{d}}{\Delta {y}_n}  
\end{align}
Inside, we define ${{\bar x}_{mn}} \buildrel \Delta \over = {{\bar x}_m} - {{\bar x}_n}$, ${{\bar y}_{mn}} \buildrel \Delta \over = {{\bar y}_m} - {{\bar y}_n}$, ${{\bar z}_{mn}} \buildrel \Delta \over = {{\bar z}_m} - {{\bar z}_n}$; 
Further replacing $\Delta z_n$ with $\Delta x_n$ and $\Delta y_n$ via \eqref{eq:DeltaZn}, and substituting it into $I_{T}$, we get
\begin{align}
	\nonumber
	I_{T} 
	\nonumber
	&= \int_{ - \frac{{l_T^h}}{2}}^{\frac{{l_T^h}}{2}} {{e^{ - i{k_0}\frac{{\left( {{{\bar x}_{mn}} + {{\bar z}_{mn}} \frac{\sin \phi_{h} \cot \theta_v - \sin \phi_v \cot \theta_h }{\sin (\phi_{h}-\phi_v)} } \right) \Delta {x_n}}}{{{{\bar d}_{mn}}}}}}{\rm{d}}\Delta {x_n}} \\
	\nonumber
	&\times \int_{ - \frac{{l_T^v}}{2}}^{\frac{{l_T^v}}{2}} {{e^{ - i{k_0}\frac{{ \left( {{\bar y}_{mn}} + {{\bar z}_{mn}} \frac{\cos \phi_h \cot \theta_v - \cos \phi_{v} \cot \theta_h }{\sin (\phi_{h}-\phi_v)} \right) \Delta {y_n}}}{{{{\bar d}_{mn}}}}}}{\rm{d}}\Delta {y_n}}. 
\end{align}
The integrals can be computed via the Euler's formula $e^{ix} = \cos x + i\sin x$ and ${\rm{sinc}} (x) = \frac{\sin x}{x}$
\begin{align}
	\nonumber
	I_{T} 
	= l_T^h l_T^v &\cdot {\rm{sinc}}\left( {\frac{{\pi l_T^h}}{\lambda} \cdot \frac{{{{\bar x}_{mn}} + {{\bar z}_{mn}} \frac{\sin \phi_{h} \cot \theta_v - \sin \phi_v \cot \theta_h}{\sin (\phi_{h}-\phi_v)} }}{{{{\bar d}_{mn}}}}} \right) \\
	\nonumber 
	&\cdot {\rm{sinc}}\left( {\frac{{\pi l_T^v}}{\lambda} \cdot \frac{{{{\bar y}_{mn}} + {{\bar z}_{mn}} \frac{\cos \phi_h \cot \theta_v - \cos \phi_{v} \cot \theta_h }{\sin (\phi_{h}-\phi_v)} }}{{{{\bar d}_{mn}}}}} \right). 
\end{align}
Employing \eqref{eq:DeltaZm} and performing a similar derivation process, we have $I_R$ directly expressed as 
\begin{align}
	\nonumber
	&I_{R} = \int_{s_R} e^{ i {k_0}  \frac{\bar{\bm{d}}_{mn}^{T}}{\bar{d}_{mn}} \Delta {{\bm{r}}_m} } {\rm{d}}{\Delta {\bm{r}}_m} \\
	\nonumber
	&= l_R^h l_R^v \cdot {\rm{sinc}}\left( {\frac{{\pi l_R^h}}{\lambda} \cdot \frac{{{{\bar x}_{mn}} + {{\bar z}_{mn}} \frac{\sin \psi_{h} \cot \vartheta_v - \sin \psi_v \cot \vartheta_h}{\sin (\psi_{h}-\psi_v)} }}{{{{\bar d}_{mn}}}}} \right) \\
	\nonumber
	&\qquad \quad \cdot {\rm{sinc}}\left( {\frac{{\pi l_R^v}}{\lambda} \cdot \frac{{{{\bar y}_{mn}} + {{\bar z}_{mn}} \frac{ \cos \psi_h \cot \vartheta_v - \cos \psi_{v} \cot \vartheta_h }{\sin (\psi_{h}-\psi_v)} }}{{{{\bar d}_{mn}}}}} \right).
\end{align}
Substituting $I_{T}$ and $I_{R}$ into \eqref{eq:MP2MP-Channel-2}, we get the channel model \eqref{eq:MP2MP-Channel-3}, which completes the proof.

\section{Proof of Lemma \ref{LemmaIntegralApprox}}
\label{LemmaIntegralApproxProof}

We first define an arbitrary point $\bm{r} \in A$. It is obvious to know that $\bm{r}$ is within the region of a certain antenna element. Without loss of generality, we assume that $\bm{r} = \bm{r}_{k}$ associates with the $k$-th antenna element, which can be represented by $\bm{r}_{k} = \bar{\bm{r}}_{k} + \Delta \bm{r}_{k}$, where $\bar{\bm{r}}_{k}$ is the element center and $\Delta \bm{r}_{k}$ belongs to area $S$. 
Then, the integral of two differentiable functions, $\bm{f}(\bm{r})$ and $\bm{g}(\bm{r})$, over the surface area $A$ can be reformulated as 
\begin{align}
	\nonumber
	\int_{A} \bm{f}^{H}(\bm{r}) \bm{g}(\bm{r}) \mathrm{d} \bm{r} 
	&= \sum_{k=1}^{K} \int_{S} \bm{f}^{H}(\bm{r}_{k}) \bm{g}(\bm{r}_{k}) \mathrm{d} \bm{r}_{k} \\
	\nonumber
	&= \sum_{k=1}^{K} \int_{S} \bm{f}^{H}(\bar{\bm{r}}_{k} + \Delta \bm{r}_{k}) \bm{g}(\bar{\bm{r}}_{k} + \Delta \bm{r}_{k}) \mathrm{d} \Delta \bm{r}_{k}, 
\end{align}
By further writing $\bm{f}( \bar{\bm{r}}_{k} + \Delta \bm{r}_{k} )$ and $\bm{g}( \bar{\bm{r}}_{k} + \Delta \bm{r}_{k} )$ via the Taylor series expansion, we have the following expressions and their approximations 
\begin{align}
	\nonumber
	\bm{f} \left( \bar{\bm{r}}_{k} + \Delta \bm{r}_{k} \right) 
	&= \bm{f} \left( \bar{\bm{r}}_{k} \right) + \bm{J}_{\bm{f}} \left( \bar{\bm{r}}_{k} \right) { \Delta \bm{r}_{k} } + o\left( {\left\| \bar{\bm{r}}_{k} \right\|_2^2} \right) 
	\approx \bm{f} \left( \bar{\bm{r}}_{k} \right), \\
	\nonumber
	\bm{g} \left( \bar{\bm{r}}_{k} + \Delta \bm{r}_{k} \right) 
	&= \bm{g} \left( \bar{\bm{r}}_{k} \right) + \bm{J}_{\bm{g}} \left( \bar{\bm{r}}_{k} \right) { \Delta \bm{r}_{k} } + o\left( {\left\| \bar{\bm{r}}_{k} \right\|_2^2} \right)
	\approx \bm{g} \left( \bar{\bm{r}}_{k} \right), 
\end{align}
where $\bm{J}_{\bm{f}} \left( \bm{r} \right)$ and $\bm{J}_{\bm{g}} \left( \bm{r} \right)$ denote the Jacobian matrices of function $\bm{f}(\bm{r})$ and $\bm{g}(\bm{r})$, respectively. The approximations hold due to the infinitesimal area of the antenna elements. 
As a consequence, we have the integral well approximated as
\begin{align}
	\nonumber
	&\int_{A} \bm{f}^{H}(\bm{r}) \bm{g}(\bm{r}) \mathrm{d} \bm{r}
	\approx \sum_{k=1}^{K} \int_{S} \bm{f}^{H}( \bar{\bm{r}}_{k} ) \bm{g}( \bar{\bm{r}}_{k} ) \mathrm{d} \Delta \bm{r}_{k} \\
	\nonumber
	&\qquad \quad = \sum_{k=1}^{K} \bm{f}^{H}( \bar{\bm{r}}_{k} ) \bm{g}( \bar{\bm{r}}_{k} ) \int_{S}  \mathrm{d} \Delta \bm{r}_{k} 
	\approx S \sum_{k=1}^{K} \bm{f}^{H}( \bar{\bm{r}}_{k} ) \bm{g}( \bar{\bm{r}}_{k} ), 
\end{align}
which completes the proof.

\section{Proof of Theorem \ref{Theorem1}}
\label{Theorem1Proof}

To transform the unknown $\gamma_{pp}$ to some known replacements, we resort to the concave property of the $\log_{2}(\cdot)$ function. As such, using the Jensen's inequality \cite{Cover2006Elements}, we get the capacity upper bounded
\begin{align}
	\nonumber
	C &= \log_{2} \left( {\prod\nolimits_{p = 1}^P {\left( {1 + \mu \cdot \mathrm{SNR} \cdot \gamma_{pp}^2} \right)} } \right) \\
	\nonumber
	&\overset{(1)}{\le} P \cdot \log_{2} \left( \frac{1}{P} \sum\limits_{p = 1}^P {\left( {1 + \mu \cdot \mathrm{SNR} \cdot \gamma_{pp}^2} \right)} \right) \\
	\nonumber
	&\overset{(2)}{=} P \cdot \log_{2} \left( 1 + \frac{\mu \mathrm{SNR}}{P} {\int_{A_{R}} {\int_{A_{T}} {{{\left\| {{\bm{G}}\left( {{{\bm{r}}_m},{{\bm{t}}_n}} \right)} \right\|}^2}{\rm{d}}{{\bm{t}}_n}} {\rm{d}}{{\bm{r}}_m}} } \right),
\end{align}
where $\overset{(1)}{\le}$ holds by using the inequality between the arithmetic mean and the geometric mean, i.e., $\prod\nolimits_{p = 1}^P x_{p} \le \left( \frac{1}{P} \sum_{p=1}^{P} x_{p} \right)^{P}$ for $x_{p} \ge 0$; $\overset{(2)}{=}$ is obtained using $\sum_{p=1}^{P} \gamma_{pp}^{2} = {\int_{A_{R}} {\int_{A_{T}} {{{\left\| {{\bm{G}}\left( {{{\bm{r}}_m},{{\bm{t}}_n}} \right)} \right\|}^2}{\rm{d}}{{\bm{t}}_n}} {\rm{d}}{{\bm{r}}_m}} }$ derived in \cite[Appendix A]{Dardari2020Communicating2} that can be dated back to \cite[Appendix B]{Piestun2000Electromagnetic}. To present a more explicit result of this upper bound, we apply the tensor Green's function \eqref{eq:GreenFunc-1} and obtain 
\begin{align}
	\nonumber
	C &\le P \cdot \log_{2} \Bigg( 1 + \frac{\mu \cdot \mathrm{SNR}}{P} \Bigg. \\
	\nonumber
	&\; \quad \Bigg. \cdot \int_{A_{R}} \int_{A_{T}} {\rm{trace}} \left[ {{\bm{G}}^{H} \left( {{{\bm{r}}_m},{{\bm{t}}_n}} \right)} {{\bm{G}}\left( {{{\bm{r}}_m},{{\bm{t}}_n}} \right)} \right] {\rm{d}}{{\bm{t}}_n} {\rm{d}}{{\bm{r}}_m} \Bigg) \\
	\nonumber
	&= P \cdot \log_{2} \Bigg( 1 + \frac{\mu \cdot \mathrm{SNR}}{P} \cdot \sum\limits_{m=1}^{M} \sum\limits_{n=1}^{N} \int_{s_{R}} \int_{s_{T}} \Bigg. \\
	\nonumber
	&\;\;\; \quad \Bigg. \qquad \qquad {\rm{trace}} \left( {{\bm{G}}^{H} \left( {{{\bm{r}}_m},{{\bm{t}}_n}} \right)} {{\bm{G}}\left( {{{\bm{r}}_m},{{\bm{t}}_n}} \right)} \right)  {\rm{d}}{{\bm{t}}_n} {\rm{d}}{{\bm{r}}_m} \Bigg) \\
	\nonumber
	&\approx P \cdot \log_{2} \Bigg( 1 + \frac{\mu \cdot \mathrm{SNR}}{P} \cdot  s_{R} s_{T} \cdot  \sum\limits_{m=1}^{M} \sum\limits_{n=1}^{N} \Bigg. \\
	\nonumber
	&\;\;\;\; \qquad \qquad \qquad \qquad  \Bigg. {\rm{trace}} \left( {{\bm{G}}^{H} \left( {{\bar{\bm{r}}_m},{\bar{\bm{t}}_n}} \right)} {{\bm{G}}\left( {{\bar{\bm{r}}_m},{\bar{\bm{t}}_n}} \right)} \right)  \Bigg) \\
	\nonumber
	&= P \cdot \log_{2} \Bigg( 1 + \frac{\mu \cdot \mathrm{SNR}}{P} \cdot s_{R} s_{T} \cdot \sum\limits_{m=1}^{M} \sum\limits_{n=1}^{N}  \Bigg. \\
	\nonumber
	&\;\;\;\; \qquad \qquad \qquad \qquad \qquad \quad \Bigg. \left( \frac{\varepsilon_1}{{\bar{d}_{mn}^2}} + \frac{\varepsilon_2}{{\bar{d}_{mn}^4}} + \frac{\varepsilon_3}{{\bar{d}_{mn}^6}} \right) \Bigg),
\end{align}
where the approximation holds for infinitesimal antenna elements, and is obtained based on Lemma \ref{LemmaIntegralApprox}; the coefficients are given by \eqref{eq:coefficients(a)} \eqref{eq:coefficients(b)} \eqref{eq:coefficients(c)}, which completes this proof.

\bibliographystyle{IEEEtran}
\bibliography{IEEEabrv,references} 

\end{document}